\pdfoutput=1
\documentclass[nonacm,screen,acmsmall,natbib]{acmart}

\usepackage[utf8]{inputenc}
\usepackage[tbtags]{mathtools}
\allowdisplaybreaks

\usepackage[capitalise,nameinlink]{cleveref}

\usepackage{tikz}
\usepackage{pgfplots}
\pgfplotsset{compat=1.14}

\newcommand{\sset}[1]{\left\{ #1\right\}}

\newcommand{\prob}[1]{\ensuremath{\mathrm{Pr}\left[#1\right]}}
\DeclareMathOperator*{\argmax}{argmax}
\newcommand{\algoname}[1]{\ensuremath{\text{\rm\sc #1}}}
\newcommand{\map}{\longrightarrow}
\newcommand{\orderstat}[3]{{#1}_{{#2}:{#3}}}
\DeclareMathOperator*{\expectation}{\mathbb E}
\newcommand{\expect}[2][]{\expectation_{#1}\nolimits\left[#2\right]}
\newcommand{\expectsmall}[2][]{\expectation_{#1}\nolimits[#2]}
\newcommand{\myerson}[2]{\algoname{Myerson}(#1,#2)}
\newcommand{\price}[2]{\algoname{Price}(#1,#2)}
\newcommand{\pricep}[3]{\algoname{Price}(#1,#2,#3)}
\newcommand{\apx}[2]{\algoname{APX}(#1,#2)}

\newtheorem*{lemmanonum}{Lemma}

\begin{document}

\title{Optimal Pricing For MHR and \texorpdfstring{$\lambda$}{lambda}-Regular Distributions}
\titlenote{An earlier version of this paper,
not including the results for $\lambda$-regular distributions, appeared in
WINE'18~\citep{gz2018}.\\ Supported by the Alexander von Humboldt Foundation with
funds from the German Federal Ministry of Education and Research (BMBF). }

\author{Yiannis Giannakopoulos}
\email{yiannis.giannakopoulos@tum.de}
\orcid{0000-0003-2382-1779}
\author{Diogo Poças}
\email{diogo.pocas@tum.de}
\orcid{0000-0002-5474-3614}
\affiliation{%
  \institution{TU Munich}
  \department{Chair of Operations Research}
  \streetaddress{Arcisstr.~21}
  \city{München}
  \postcode{80333}
  \state{Bayern}
  \country{Germany}
}

\author{Keyu Zhu}
\authornote{Work done mostly while visiting the Chair of Operations Research of TU Munich.}
\email{keyu.zhu@gatech.edu}
\affiliation{%
  \institution{Georgia Institute of Technology}
  \department{School of Industrial and Systems Engineering}
  \streetaddress{765 Ferst Drive NW}
  \city{Atlanta}
  \state{GA}
  \postcode{30332}
  \country{USA}
  }

\begin{abstract}
We study the performance of anonymous posted-price selling mechanisms for a
standard Bayesian auction setting, where $n$ bidders have i.i.d.\ valuations for a
single item. We show that for the natural class of Monotone Hazard Rate (MHR)
distributions, offering the same, take-it-or-leave-it price to all bidders can
achieve an (asymptotically) optimal revenue. In particular, the approximation
ratio is shown to be $1+O(\ln \ln n/\ln n)$, matched by a tight lower bound
for the case of exponential distributions. This improves upon the
previously best-known upper bound of $e/(e-1)\approx 1.58$ for the slightly more
general class of regular distributions. In the worst case (over $n$), we still show
a global upper bound of $1.35$. We give a simple, closed-form description
of our prices which, interestingly enough, relies only on minimal knowledge of the
prior distribution, namely just the expectation of its second-highest order
statistic.

Furthermore, we extend our techniques to handle the more general class of
$\lambda$-regular distributions that interpolate between MHR ($\lambda=0$) and
regular ($\lambda=1$). Our anonymous pricing rule now results in an asymptotic
approximation ratio that ranges smoothly, with respect to $\lambda$, from $1$ (MHR
distributions) to $e/(e-1)$ (regular distributions). Finally, we explicitly give a
class of continuous distributions that provide matching lower bounds, for every
$\lambda$.
\end{abstract}



\maketitle

\section{Introduction} In this paper we study a traditional Myersonian auction
setting: an auctioneer has an item to sell and he is facing $n$ potential buyers.
Each buyer has a (private) valuation for the item, and these valuations are i.i.d.\
according to some known continuous probability distribution $F$. You can think of
this valuation, as modelling the amount of money that the buyer is willing to spend
in order to get the item. An auction is a mechanism that receives as input a bid
from each buyer, and then decides if the item is going to be sold and to whom, and
for what price. Our goal is to design auctions that maximize the seller's expected
revenue.

We focus only on truthful auctions, that is, selling mechanisms that give no
incentives to the bidders to lie about their true valuation. Such auctions are both
conceptually and practically convenient. This restriction is essentially without
loss for our revenue maximization objective, due to the Revelation
Principle\footnote{In this paper we will avoid discussing such subtler issues as
implementability and truthfulness, since our goal is to study the performance of
specific and very simple pricing mechanisms. The interested reader is pointed to
\citep{Nisan2007a} as a good starting point for a deeper investigation of those
ideas.}.

In general, such an optimal auction can be rather complicated and even randomized
(aka a lottery). However, in his celebrated result, \citet{Myerson1981a} proved that
(under some standard assumptions on the valuations' distribution) revenue
maximization can be achieved by a very simple deterministic mechanism, namely a
second-price auction paired with a reserve value $r$. In such an auction, all buyers
with bids smaller than $r$ are ignored and the item is sold to the highest bidder
for a price equal to the second-highest bid (or $r$, if no other bidder remains).
Equivalently, you can think of this as the seller himself taking part in the
auction, with a  bid equal to $r$, and simply running a standard, Vickrey
second-price auction; if the auctioneer is the winning bidder, then the item stays
with him, that is, it remains unsold. Furthermore, \Citet{Bulow1996} essentially
showed that we can still guarantee a $1-\frac{1}{n}$ fraction of this optimal
revenue, even if we drop the reserve price $r$ completely and use just a standard
second-price auction.

No matter how simple and powerful the above optimal auction seems, it still requires
explicitly soliciting bids from all buyers and using the second-highest as the
``critical payment''; this is essentially a centralized solution, that asks for a
certain degree of coordination. Arguably, there is an even simpler selling mechanism
which, as a matter of fact, is being used extensively in practice, known as
\emph{anonymous pricing}: the seller simply decides on a selling price $p$, and then
the item goes to any buyer that can afford it (breaking ties arbitrarily); that is,
we sell the item to any bidder with a valuation greater or equal to $p$, for a price
of exactly $p$.

The question we investigate in this paper, is how well can such an extremely simple
selling mechanism perform when compared to an arbitrary, optimal auction. We resolve
this in a very positive way proving that, under natural assumptions on the valuation
distribution, as the number of buyers grows large, anonymous pricing achieves
optimal revenue. More precisely, its approximation ratio is $1+O(\ln\ln n/\ln n)$.
Furthermore, we show that in order to get such a near-optimal performance, the
seller does not really need to have full knowledge of the bidders' population; he
just needs to know the expectation of the second-highest order statistic of the
valuation distribution, that is, (a good estimate of) the \emph{expected}
second-highest bid is enough. Finally, we demonstrate how this approximation ratio
deteriorates as we gradually relax our distributional assumptions.

\subsection{Related Work}
\label{sec:related}
The seminal reference in auction theory is the work of \citet{Myerson1981a} who
completely characterized the revenue-maximizing auction in single-item settings with
bidder valuations drawn from independent (but not necessarily identical)
distributions. Under his standard regularity condition (see
\cref{sec:lambdareg-model}), this optimal auction has a very simple description when
the valuation distributions are identical: it is a second-price auction with a
reserve. Furthermore, there is an elegant, closed-form formula that gives the
reserve price (see \cref{sec:model}).

One can achieve good, constant approximations to that optimal revenue by using even
simpler auctions, namely \emph{anonymous pricing} mechanisms. These mechanisms offer
the same take-it-or-leave-it price to all bidders, and the item is sold to someone
who can afford it (breaking ties arbitrarily). An upper bound of $e/(e-1)\approx
1.58$ on the approximation ratio of anonymous pricing can be shown from the work of
\citet{Chawla2010a}. 
\citet{Blumrosen2008a} study the asymptotic performance of
pricing when the number of bidders grows large and demonstrate a lower bound on the
approximation ratio of $0.88/0.65=1.37$ for anonymous pricing.
If we allow for non-continuous distributions that have point-masses, then
\citet{Duetting2016} provide a matching lower bound of $e/(e-1)$. Although the class
of MHR distributions (see \cref{sec:lambdareg-model}) is a natural restriction of
Myerson's regularity, that has been extensively studied in optimal auction theory,
mechanism design and complexity to derive powerful positive results (see, e.g.,
\citep{Hartline2009a,Bhattacharya2010a,Dhangwatnotai2014a,babaioff2017posting,Daskalakis2012b,Cai2011b,gkyr2015,gkl2017}),
no better bounds are known for anonymous pricing in this class. This is one of our
goals in this paper.

\Citet{SchweizerSzech2019} propose a quantitative notion of regularity, termed
$\lambda$-regularity (see \cref{sec:lambdareg-model}), that allows for a smooth
interpolation between the general class of regularity à la Myerson and its MHR
restriction. This is essentially equivalent to the notion of $\alpha$-strong
regularity of \citet{Cole2014a} (for $\alpha=\lambda -1$) and $\rho$-concavity of
\citet{Caplin:1991aa} (for $\rho=-\lambda$). These parametrizations have proved very
useful in developing a fruitful and more ``fine-grained'' theory of optimal auctions
(see, e.g., \citep{Cole2017,Ewerhart:2013aa,Mares:2011aa,Mares:2014aa}).

Although not immediately related to our model, an important line of work studies the
performance of ``simple'' auctions, such as pricing and auctions with reserves, for
the more general case where bidders' valuations may be non-identically distributed.
In such settings, the elegance of Myerson's characterization is not in effect any
more, and the optimal auction can be rather complicated. Nevertheless, in an
influential paper, \citet{Hartline2009a} showed that, for regular distributions, a
second-price auction with a single anonymous reserve guarantees a $4$-approximation
to the optimal ratio, and also provided a lower bound of $2$. This upper bound was
subsequently improved to $e\approx 2.72$ by \citet{Alaei2015}, achieved even by the
simpler class of anonymous pricing mechanisms. At the same paper, they also provided
a lower bound of $2.23$ for the approximation ratio of anonymous pricing for
non-i.i.d.\ bidders. This was recently improved to $2.62$~\citep{Jin2018}
and proven to be tight by~\citet{Jin:2019aa}.
For bounds on the approximation ratios between different
pricing and reserve mechanisms, under various assumptions on the underlying
distributions and the order of the bidders' arrival,
see~\citep{Alaei2015,Duetting2016,Jin2018,Chawla2010a} and \citep[Chapter~4]{Hartlinea}.
For anonymous pricing beyond the standard setting of linear utilities see the very recent work of~\citet{Feng:2019aa}. 

Finally, we briefly mention that there is a very rich theory about sequential
pricing that deals with dynamically arriving buyers and which is inspired by and
related to secretary-like online problems and the powerful theory of prophet
inequalities. See, e.g.,
\citep{Hajiaghayi2007a,Kleinberg:2012,Alaei2014,Yan2011,Cesa-Bianchi2015,Chawla2010a,Correa:2019aa}. An intriguing equivalence between pricing and threshold stopping rules was recently established by~\citet{Correa_2019}. 
In particular, \citet{Correa2017} showed an upper bound of $1.34$ on the approximation ratio of sequential pricing in the i.i.d.\ auction model which translates to the same bound for the i.i.d.\ prophet setting; this resolved a long-standing open question from~\citet{Hill:1982aa}.
Similarly, it is not difficult to see that our setting of regular i.i.d.\ anonymous pricing corresponds to single-threshold rules for the i.i.d.\ prophet model, for which a tight bound of $1.58$ is known~\citep{Hill:1982aa,Ehsani:2018aa}.

\subsection{Our Results}

In this paper we study the performance of anonymous pricing mechanisms in
single-item auction settings with $n$ bidders that have i.i.d.\ valuations from the
same regular distribution $F$. These mechanisms are extremely simple: the seller
simply offers the same take-it-or-leave-it price $p$ to all potential buyers; the
item is then sold to a buyer that can meet this price, that is, has a valuation
greater than or equal to $p$; the winning bidder pays $p$ to the seller. Our benchmark
is the seller's expected revenue (with respect to his incomplete, prior knowledge of
the buyers' bids via distribution $F$) and we compare against the maximum revenue
achievable by any auction. For our particular model, this optimal auction is a
second-price auction with a reserve~\citep{Myerson1981a}.

Our main result (\cref{sec:upper}; see also \cref{fig:approx_upper}) is an explicit,
closed-form upper bound on the approximation ratio of the revenue of anonymous
pricing for MHR distributions. As the number $n$ of buyers grows large, this ratio
tends to the optimal value of $1$, at a rate of $1+O(\ln\ln n/\ln n)$
(\cref{th:main_upper}). Additionally, we design an upper bound that is fine-tuned to
handle also small values of $n$ (\cref{th:main_upper_small}), and using this we
provide a global, worst-case (with respect to $n$) upper bound of $1.35$ on the
approximation ratio. Previously, only an upper bound of $e/(e-1)\approx 1.58$ was
known (for any value of $n$), holding for the entire class of regular distributions.

In \cref{sec:prior_indi} we demonstrate how the aforementioned positive guarantee on
the revenue of anonymous pricing can still be (within an exponentially decreasing
\emph{additive} constant) achieved even if the seller does not have full knowledge
of the prior distribution $F$ (see \cref{fig:approx_lower}). In particular
(\cref{cor:upper_bound_prior_indi}), we give an explicit formula for such a ``good''
pricing rule that only depends on the expectation of the second-highest order
statistic of $F$.

To complete the picture, in \cref{sec:lower} we prove that our upper bound analysis
is essentially tight, by showing that the exponential distribution provides an
(almost) tight gap instance between the revenue of anonymous pricing and that of the
optimal auction (\cref{th:lower_expo}; see also \cref{fig:approx_lower}).

Finally, in~\cref{sec:lambdareg} we relax our MHR assumption, allowing for
$\lambda$-regular valuation distributions. This provides a smooth, parametrized
generalization of the MHR condition ($\lambda=0$), all the way to the entire class
of standard (Myersonian) regularity ($\lambda=1$). Extending our ideas
from~\cref{sec:MHR}, we are able to provide upper (\cref{thm:regularuppfix}) and
lower (\cref{thm:regularlowfix}) bounds on the approximation ratio of anonymous
pricing, for all $\lambda\in[0,1]$ (see also \cref{fig:regupplowbounds}). We
conclude in \cref{sec:lambda_reg_asy}, by looking how these bounds behave in the
limit as the number of bidders grows arbitrarily large; we derive a tight value for
the approximation ratio as a function of $\lambda$, which ranges smoothly from
optimality for MHR distributions ($\lambda=0$) to $e/(e-1)\approx 1.58$ for the most
general case of regular distributions ($\lambda=1$).

We conclude in~\cref{sec:conclusion} with some open questions for future work.

\subsubsection{Techniques}

Our upper bound technique differs from related previous
approaches~\cite{Chawla2010a,Alaei2015} in that we do not use the ex-ante relaxation
of the revenue-maximization objective. Instead, we deploy explicit upper bounds on
the optimal revenue (\cref{sec:bounds_optimal_rev}) that depend on key parameters of
the valuation distribution $F$, namely its order statistics and its monopoly
reserve. Then, we pair these with a range of critical properties of
$\lambda$-regular distributions that we develop in \cref{sec:MHR_props,sec:lambda_reg-props}. We believe
that some of these auxiliary results may be of independent interest, in particular
the order statistics tail-bounds
of~\cref{lemma:bound_tail_orderstats,lemma:bound_tail_regular} and the
reserve-quantile optimal revenue bound of \cref{lem:bounds_reserve_price} for the
special case of MHR distributions.

Our lower bounds are constructed by means of explicitly defining a family of
$\lambda$-regular, continuous valuation distributions
(see~\eqref{eq:pareto_instances}), that act as ``bad'' instances for any $\lambda$.
We achieve this by generalizing an instance by~\citet{Duetting2016} (tailored to the
special case of regular distributions) to work for general $\lambda$'s, while at the
same time ``smoothing it out'' to satisfy continuity.

\section{Model and Notation}
\label{sec:model}

A seller wants to sell a single item to $n\geq 2$ bidders. The valuations of the
bidders for the item are i.i.d.\ from a continuous probability distribution
supported over an interval $D_F\subseteq [0,\infty)$, with cdf $F$ and pdf $f$.
Throughout this paper we will assume that $F$ is $\lambda$-regular, for some real parameter $\lambda\in[0,1]$ (for formal definitions and discussion, see~\cref{sec:lambdareg-model} right below).
For a random variable $X\sim F$ drawn from $F$ and $1\leq k\leq n$, we will use
$\orderstat{X}{k}{n}$ to denote the $k$-th lowest order statistic out of $n$ i.i.d.\
draws from $F$. That is, $\orderstat{X}{1}{n}\leq \orderstat{X}{2}{n}\leq\dots\leq
\orderstat{X}{n}{n}$. For completeness and ease of reference, we discuss some useful
properties of order statistics in \cref{append:prob_basics}. 

A pricing mechanism that offers a take-it-or-leave-it price of $p\in D_F$ to all
bidders gives to the seller an expected revenue of
$$
\pricep{F}{n}{p}\equiv p[1-F^n(p)]\;,
$$
since the probability of no bidder being able to afford price $p$ is $F^n(p)$. We
will refer to such a mechanism simply as \emph{(anonymous) pricing}. Thus, the
optimal (maximum) revenue achievable via pricing is
$$
\price{F}{n}\equiv\sup_{p\in D_F}\pricep{F}{n}{p}\;.
$$
On the other hand, as discussed in the introduction, the optimal revenue attainable
by any mechanism may be higher; as a matter of fact, \citet{Myerson1981a} showed
that it is achieved by a second-price auction with a reserve equal to the
\emph{monopoly reserve}\footnote{As will become clearer in the
following~\cref{sec:lambdareg-model}, this quantity is well-defined for all
$\lambda$-regular distributions with $\lambda\in[0,1)$. For the special case of
$\lambda=1$, it might happen that there are multiple maximizers
in~\eqref{eq:monopoly_price}, or even none. To formally deal with such pathological
cases, in the former we can take $r^*=\inf\argmax_{r\in D_F} r(1-F(r))$. In the
latter we can slightly abuse notation and use $r^*=\infty$, which is essentially
equivalent to using an arbitrarily large price to approximate within arbitrary
accuracy the value of $\sup_{r\in D_F} r(1-F(r))$; then, the corresponding reserve
quantile is $q^*=0$ and maximizes the revenue curve in~\eqref{eq:revenue_curve} (for a
more detailed and rigorous discussion of this issue see, e.g., \citep[Appendix~1]{Fu2015}
and \citep[Appendix~C]{gk2014}.) The aforementioned choices do not affect the
validity of any of the results in the rest of our paper.}
\begin{equation}
\label{eq:monopoly_price}
r^*\equiv\argmax_{r\in D_F} r(1-F(r))\;. 
\end{equation}
of the valuation
distribution. We denote this optimum revenue by $\myerson{F}{n}$, and it can be
shown that
$$
\myerson{F}{n}\equiv\expect{\max\sset{0,\phi(\orderstat{X}{n}{n})}}\;,
$$
where $\phi(x)=x-\frac{1-F(x)}{f(x)}$ is the (nondecreasing) virtual valuation
function of $F$ (see \cref{sec:lambdareg-model}) and $\orderstat{X}{n}{n}$ its maximum
order statistic. Keep in mind that, due to the monotonicity of $\phi$ and the
definition of the reserve $r^*$, we know that $\phi(x)\geq 0$ for all $x\geq r^*$.

Sometimes it is more convenient to work in quantile space instead of the actual
valuation domain. More precisely, the quantile of distribution $F$ corresponding to
a value $x\in D_F$ is $q(x)=1-F(x)$. Using this, we can define what is known as the
\emph{revenue curve} of distribution $F$, by 
\begin{equation}
\label{eq:revenue_curve}
R(q)\equiv F^{-1}(1-q)\cdot q\;.
\end{equation} 
In other
words, if $p\in D_F$ is a price and $q$ is its corresponding quantile, then $R(q)$
is the expected revenue of selling the item to a single bidder, using a price $p$.
Thus, the monopoly reserve quantile $q^*$ that corresponds to the monopoly reserve
$r^*$ defined in~\eqref{eq:monopoly_price} is exactly a maximizer of the revenue curve $R(q)$. So, for a
single bidder $(n=1)$: $$\myerson{F}{1}=\price{F}{1}=\sup_{p\in
D_F}p(1-F(p))=\sup_{q\in [0,1]} R(q)=R(q^*)\;.$$ In general though for more players
($n\geq 2$) this is not the case, and our goal in this paper is exactly to study how
well the optimal revenue $\myerson{F}{n}$ can be approximated by pricing
$\price{F}{n}$. That is, we want to bound the following approximation ratio:
$$
\apx{F}{n}\equiv\frac{\myerson{F}{n}}{\price{F}{n}}\;.
$$

Finally, we use $H_n$ to denote the $n$-th
harmonic number $H_n=\sum_{i=1}^n\frac{1}{i}$, $\gamma\approx 0.577$ for the
Euler-Mascheroni constant (see also \cref{lem:harmonic_bounds_sequence}) and $\Gamma$, $\mathrm{B}$ for the standard gamma and beta functions: $\Gamma(x)=\int_0^{\infty}t^{x-1}e^t\,dt$, 
$\mathrm{B}(x,y)=\int_0^1 t^{x-1} (1-t)^{y-1}\,dt$. 
Recall that $\Gamma(n+1)=n!$ for any nonnegative integer $n$ and $\mathrm{B}(x,y)=\frac{\Gamma(x)\Gamma(y)}{\Gamma(x+y)}$  for all reals $x,y>0$ (see, e.g., \citep[Chapter~5]{olveretal:10}).

\subsection{\texorpdfstring{$\lambda$}{lambda}-Regular Distributions}
\label{sec:lambdareg-model}

Consider a continuous distribution $F$ supported over an interval $D_F$ of
nonnegative reals, and a real parameter $\lambda\in[0,1]$. We will say that $F$ is
\emph{$\lambda$-regular} if its \emph{$\lambda$-generalized virtual valuation}
function $$\phi_\lambda(x) \equiv \lambda \cdot x-\frac{1-F(x)}{f(x)}$$ is
monotonically nondecreasing in $D_F$. From~\citet[Proposition~1(iii)]{SchweizerSzech2019}, this can be equivalently restated as
$$
\begin{cases}
[1-F(x)]^{-\lambda}\;\;\text{is convex}, &\text{if}\;\; \lambda\in (0,1]\;,\\
\ln[1-F(x)]\;\;\text{is concave}, &\text{if}\;\;\lambda=0\;.
\end{cases}
$$ 
We will use $\mathcal{D}_{\lambda}$ to denote the family of all $\lambda$-regular
distributions. It is not difficult to see that, for any $0\leq\lambda \leq
\lambda'\leq 1$, any $\lambda$-regular distribution is also $\lambda'$-regular. In
other words, $\mathcal{D}_{\lambda}\subseteq\mathcal{D}_{\lambda'}$. As a matter of fact, this hierarchy is strict (as demonstrated by the distributions defined in~\eqref{eq:pareto_instances}). 

For the special case of $\lambda=1$, the above definition recovers exactly the
standard notion of regularity à la~\citet{Myerson1981a}. Based on this, for
simplicity and consistency, we will feel free in such cases to drop the $\lambda=1$
subscripts and refer to the corresponding distributions simply as \emph{regular} and
to the function $\phi_1(x)=\phi(x)$ as the \emph{virtual valuation}. An equivalent
way of looking at regularity, is that the revenue curve defined
in~\eqref{eq:revenue_curve} must be concave.

On the other extreme of the range, for $\lambda=0$ we get the definition of
\emph{Monotone Hazard Rate (MHR)} distributions.\footnote{A standard reference for
MHR distributions is that from~\citet{Barlow1964}. For an in-depth treatment of the
subject, we refer to the book of~\citet[Chapter~2]{Barlow1996}.} Intuitively, MHR
distributions have exponentially decreasing tails. Although they represent the
strictest class within the $\lambda$-regularity hierarchy, they are still general
enough to give rise to a wide family of natural distributions, like the uniform,
exponential, normal and gamma.

Given the above discussion, $\lambda$-regularity can be seen as a quantitative
measure of the ``regularity'' of the distribution, interpolating smoothly between
MHR ($\lambda=0$) and Myerson-regular ($\lambda=1$) distributions. It will also be
convenient to define, for each $\lambda\in[0,1]$, the worst-case approximation ratio
among all $\lambda$-regular distributions:

\[\apx{n}{\lambda}\equiv\sup_{F\in\mathcal{D}_{\lambda}}\apx{F}{n}=\sup_{F\in\mathcal{D}_{\lambda}}\frac{\myerson{F}{n}}{\price{F}{n}}\;.\]
This will be the main quantity of interest throughout our paper.

\section{Preliminaries}
\subsection{Bounds on the Optimal Revenue}
\label{sec:bounds_optimal_rev}
In this section we collect the bounds on the
optimal revenue $\myerson{F}{n}$ that we will use for our main positive results in
\cref{sec:upper,sec:lambdareg} to bound the approximation ratio of pricing. They rely on the regularity of the valuation distribution. The first one is essentially a refinement of
the well-known Bulow-Klemperer bound \cite{Bulow1996} (see, e.g., \citep[Corollary~5.3]{Hartlinea}), and it was proven by
\citet{Fu2015}:
\begin{lemma}[\citet{Fu2015}]
\label{lemma:opt_bound_Fu_et_al} For $n$ bidders with i.i.d.\ values from a
regular distribution $F$,
$$
\myerson{F}{n}\leq \expectsmall{\orderstat{X}{n-1}{n}} + R(q^*)(1-q^*)^{n-1}\;,
$$
where $X\sim F$ and $R$ is the revenue curve of $F$ and $q^*$ is the quantile
corresponding to the monopoly reserve price $r^*$ of $F$, $q^*=1-F(r^*)$.
\end{lemma}
As an immediate consequence, the bound above holds also for the more restricted classes of
$\lambda$-regular distributions, for any $\lambda\in [0,1]$, and in particular for
MHR distributions (see also the discussion in \cref{sec:model}).

Our second bound on the optimal revenue, designed particularly for MHR
distributions, is a new one and might be of independent interest also for future
work:
\begin{lemma}
\label{lem:opt_bound_2} For every MHR distribution $F$ with monopoly reserve price
$r^*$ and quantile $q^*=1-F(r^*)$, and any positive integer $n$,
$$
\myerson{F}{n} \leq  r^*\int_{0}^{q^*}\frac{1-(1-z)^{n}}{z}\; dz\;.
$$
\end{lemma}
\begin{proof} Fix an MHR distribution $F$, with monopoly reserve price $r^*$ and
corresponding quantile $q^*$. Then, we know that for the virtual valuation (see
\cref{sec:model}) it is $$\phi(r^*)=r^*-\frac{1-F(r^*)}{f(r^*)}\geq 0\;.$$ Also, from
the MHR condition, for any $x\geq r^*$ it must be that $$\frac{f(r^*)}{1-F(r^*)} \leq
\frac{f(x)}{1-F(x)}\;.$$ Combining the above we get that, for all $x\geq r^*$,
\begin{equation}
\label{eq:helper_bound_reserve} 1-F(x)\leq f(x)\cdot \frac{1-F(r^*)}{f(r^*)} \leq
r^* f(x)\;.
\end{equation}
Fix also a positive integer $n$, and let $X\sim F$. Then (see also
\cref{append:prob_basics}) the maximum order statistic $\orderstat{X}{n}{n}$ is
distributed according to $F^n$. Observe that (see also \cref{sec:model})
$$
\frac{\partial\, \pricep{F}{n}{x}}{\partial\, x}=\frac{\partial\,
x(1-F^n(x))}{\partial\, x}=1-F^n(x)-nxF^{n-1}(x)f(x)
$$
and thus
\begin{align*}
\myerson{F}{n}
    &= \expect{\max\sset{0,\phi(\orderstat{X}{n}{n})}}\\
    &= \int_{r^*}^{\infty} \phi(x)\, dF^n(x)\\
    &= \int_{r^*}^{\infty} \left(x-\frac{1-F(x)}{f(x)}\right)n F^{n-1}(x) f(x)\, dx\\
    &= \int_{r^*}^{\infty} nxF^{n-1}(x)f(x)-n(1-F(x))F^{n-1}(x)  \, dx\\
    &= \int_{r^*}^{\infty} - \frac{\partial\, \pricep{F}{n}{x}}{\partial\, x} + 1-F^n(x)-n(1-F(x))F^{n-1}(x) \, dx\\
    &= \pricep{F}{n}{r^*}+ \int_{r^*}^{\infty} 1+(n-1)F^n(x)-nF^{n-1}(x) \, dx\;.
\end{align*}
Also, due to \eqref{eq:helper_bound_reserve}, for all $x\geq r^*$:
\begin{align*}
1-F^n(x)-n(1-F(x))F^{n-1}(x)
  &=\left[\frac{1-F^n(x)}{1-F(x)}-nF^{n-1}(x) \right] (1-F(x))\\
  &\leq r^* \left[\frac{1-F^n(x)}{1-F(x)}-nF^{n-1}(x) \right] f(x)\;.
\end{align*}

By performing a change of variable to quantile space, that is setting $z=1-F(x)$ and
observing that $\frac{d\, z}{d\, x}=-f(x)$ and $1-F(r^*)=q^*$, we get that
\begin{align*}
\int_{r^*}^{\infty} 1-F^n(x)-n(1-F(x))F^{n-1}(x) \, dx
    &\leq r^*\int_0^{q^*}\frac{1-(1-z)^n}{z}-n(1-z)^{n-1}\, dz\\
    &= -r^*\left[1-(1-q^*)^n\right]+r^*\int_0^{q^*}\frac{1-(1-z)^n}{z}\, dz\\
    &= -\pricep{F}{n}{r^*}+r^*\int_0^{q^*}\frac{1-(1-z)^n}{z}\, dz\;.
\end{align*}
Thus,
$$
\myerson{F}{n} \leq r^*\int_0^{q^*}\frac{1-(1-z)^n}{z}\, dz\;.
$$\end{proof}

\subsection{MHR Distributions}
\label{sec:MHR_props} 
In this section we state some properties of MHR distributions that
will play a critical role into deriving our main results in the rest of the paper.
\Cref{lemma:bound_tail_orderstats}, in particular, might be of independent
interest, since it is providing powerful tail-bounds with respect to the order
statistics of the distribution:

\begin{lemma}
\label{lemma:bound_tail_orderstats} For any continuous MHR random variable $X$,
integers $1\leq k \leq n$ and real $c\in[0,1]$,
$$
\prob{X < c\cdot\expectsmall{\orderstat{X}{k}{n}}}\leq 1- e^{-c(H_{n}-H_{n-k})}\;.
$$
\end{lemma}
\begin{proof}\
Let $E$ denote an exponential random variable and also let $F$ and $G$ denote the cumulative probability
functions of $X$ and $\orderstat{X}{k}{n}$, respectively. Then (see also
\cref{append:prob_basics})
\begin{equation}
\label{eq:pdf_ordstat_1}
\frac{dG(y)}{dy}=k\binom{n}{k}F^{k-1}(y)(1-F(y))^{n-k} \frac{dF(y)}{dy}\;,
\end{equation} for almost all $y\in [0,\infty)$. To simplify notation, also let $\nu =
\expectsmall{\orderstat{X}{k}{n}}$. Our goal then is to upper-bound
$F(c\nu)$, i.e.\ lower-bound $1-F(c\nu)$. Since $F$ is MHR, $\zeta(c)=\ln (1-F(c\nu))$ is a
concave function of $c$ with $\zeta(0)=0$; thus, assuming $c\in[0,1]$ and applying Jensen's inequality,
\begin{equation*}
\ln (1-F(c\nu)) \geq c\ln(1-F(\nu))\geq c\expect{\ln(1-F(\orderstat{X}{k}{n}))}=c\int_0^\infty \ln(1-F(y))\,d G(y)\;;
\end{equation*}
the integral can be further simplified using \eqref{eq:pdf_ordstat_1} and the changes of variables $u=F(y)$, $t=-\ln(1-u)$;
\begin{align*}
c\int_0^\infty \ln(1-F(y))\,d G(y)
    &=c k\binom{n}{k} \int_0^\infty \ln(1-F(y)) F^{k-1}(y)(1-F(y))^{n-k} \,d F(y)\\
    &=c k\binom{n}{k} \int_0^1 \ln(1-u) u^{k-1}(1-u)^{n-k} \,d u\\
    &=-c k\binom{n}{k} \int_0^\infty t (1-e^{-t})^{k-1}(e^{-t})^{n-k} \,d t\\
    &= -c \expect{\orderstat{E}{k}{n}}\\
    &= -c(H_{n}-H_{n-k})\;,
\end{align*}
the last equality following from
\cref{append:prob_basics}. So, applying the exponential function on both sides,
finally we get the desired
$$
1-F(c\nu) \geq e^{-c(H_{n}-H_{n-k})}\;.
$$
\end{proof}

The next lemma states some useful bounds on the monopoly reserve of an MHR distribution:
\begin{lemma}
\label{lem:bounds_reserve_price}
For any MHR distribution with expectation $\mu$, monopoly reserve $r^*$ and
corresponding quantile $q^*$:
\begin{enumerate}
\item $q^*\geq 1/e$;
\item $\ln(1/q^*)\cdot \mu \leq r^* \leq \frac{\ln(1/q^*)}{1-q^*}\cdot \mu$.
\end{enumerate}
\end{lemma}
\begin{proof}
The first property is by now almost folklore, see e.g.~\citep[Lemma~1]{aggarwal2009efficiency}
or \citep[Claim~B.2]{Babaioff2015}. Alternatively, it can be seen as the limiting case of \cref{lemma:regularproperties} as $\lambda\to 0$, since $\lim_{\lambda\to 0^+}(1-\lambda)^{1/\lambda}=1/e$.

For the second property, applying \citep[Corollary 3.10]{Barlow1964} for the
first-order moments ($r=1$), we get that for any quantile $q=1-F(x)$ of our MHR
distribution with $q\geq 1/e$, it must be that
$$
-\ln q \cdot \mu \leq x \leq \left[\int_0^1 q^y\,dy\right]^{-1}\cdot \mu =
\left(\frac{q-1}{\ln  q}\right)^{-1} \mu\;.
$$
By the first property of our lemma, it is valid to use the above inequality with
the monopoly reserve quantile $q^*=1-F(r^*)$ and so, by setting $q\gets q^*$ and
$x\gets r^*$ we have that
$$
-\ln(q^*)\mu \leq r^* \leq -\frac{\ln(q^*)}{1-q^*}\mu\;.
$$
\end{proof}

Finally, the following lemma shows that the high-order statistics of MHR
distributions are ``well-behaved'', in the sense that they cannot be away from the
expectation:
\begin{lemma}
\label{lem:bound_2ndordstat_mean_MHR} For any MHR random variable $X$ and integer
$n\geq 2$,
$$
\expectsmall{\orderstat{X}{n-1}{n}}\geq \left(1-\frac{H_n-1}{n-1}\right)\cdot
\expectsmall{X}\;.
$$
\end{lemma}
\begin{proof} For convenience denote $\mu=\expectsmall{X}$ and
$\nu=\expectsmall{\orderstat{X}{n-1}{n}}$. From
\citet[Lemma~5.3]{babaioff2017posting} we know that, since $X$ is MHR, its
highest-order statistic is upper bounded by
$$
\expectsmall{\orderstat{X}{n}{n}}\leq H_n\cdot \mu\;.
$$
Using this we get that:
$$
n\cdot\mu = \expect{\sum_{i=1}^n
\orderstat{X}{i}{n}}=\sum_{i=1}^{n-1}\expectsmall{\orderstat{X}{i}{n}}+\expectsmall{\orderstat{X}{n}{n}}\leq
(n-1)\nu + H_n\cdot \mu\;,
$$
and thus $(n-1)\nu \geq (n-H_n)\mu$, or equivalently, $\nu \geq
\frac{(n-H_n)}{n-1}\mu=\left(1-\frac{H_n-1}{n-1}\right)\mu$.
\end{proof}

\subsection{\texorpdfstring{$\lambda$}{lambda}-Regular Distributions}
\label{sec:lambda_reg-props}

The following lemmas are the counterparts of \cref{lemma:bound_tail_orderstats} and
\cref{lem:bounds_reserve_price} (Property~1), extending them to $\lambda$-regular
distributions for $\lambda>0$.

\begin{lemma}
\label{lemma:bound_tail_regular}
Let $X$ be a $\lambda$-regular distribution, for $\lambda\in(0,1]$. Then, for any
integer $1\leq k\leq n$ and real $c\in[0,1]$,
$$
\prob{X \leq c\cdot\expectsmall{\orderstat{X}{k}{n}}}\leq 1-
\left(1+c\left(\frac{n!\Gamma(n+1-k-\lambda)}{(n-k)!\Gamma(n+1-\lambda)}-1\right)\right)^{-1/\lambda}\;.
$$
\end{lemma}

\begin{proof}
Let $F$ and $G$ be the cumulative density functions of $X$ and $\orderstat{X}{k}{n}$
respectively. Let also $\nu=\expectsmall{\orderstat{X}{k}{n}}$. As $F$ is
$\lambda$-regular, the function $\zeta(c)=(1-F(c\nu))^{-\lambda}$ is a convex
function on $c$, with $\zeta(0)=1$; thus, assuming $c\in[0,1]$, and applying
Jensen's inequality,

\begin{equation}
\zeta(c)\leq 1+c(\zeta(1)-1)=1+c((1-F(\nu))^{-\lambda}-1)\leq 1+c(\expectsmall{(1-F(\orderstat{X}{k}{n}))^{-\lambda}}-1)\;.
\label{eq:zeta}\end{equation}

The rest of the proof follows exactly as in the proof of
\cref{lemma:bound_tail_orderstats}. We use a change of variable to compute the
expected value:

\begin{align*}
\expectsmall{(1-F(\orderstat{X}{k}{n}))^{-\lambda}}
&=\int_0^\infty (1-F(y))^{-\lambda}dG(y)\\
&=k\binom{n}{k}\int_0^\infty F(y)^{k-1}(1-F(y))^{n-k-\lambda}dF(y)\\
&=k\binom{n}{k}\int_0^1 u^{k-1}(1-u)^{n-k-\lambda}du\\
&= k\binom{n}{k} \mathrm{B}(k,n+1-k-\lambda)\\
&=\frac{n!\Gamma(n+1-k-\lambda)}{(n-k)!\Gamma(n+1-\lambda)}\;.
\end{align*}
Plugging
this into \eqref{eq:zeta} and rearranging gives us the desired result.
\end{proof}

\begin{lemma}[\citet{SchweizerSzech2019}]
\label{lemma:regularproperties}
For any $\lambda$-regular distribution with $\lambda\in(0,1]$ and monopoly quantile $q^\ast$,
\begin{equation*}
q^\ast\geq(1-\lambda)^{1/\lambda}\;.
\end{equation*}
\end{lemma}

\section{Bounds for MHR Distributions}
\label{sec:MHR}

To facilitate us with stating and proving our bounds for the approximation
ratio of pricing, we define the following auxiliary function $g_n:[0,\infty)\map
[0,\infty)$, for any positive integer $n$,
\begin{equation}
\label{eq:gdef} g_n(c) \equiv c[1-(1-e^{-c(H_n-1)})^n]
\end{equation} 
and its (unique) maximizer in $[0,1]$ by
\begin{equation}
\label{eq:cparsdef} c_n \equiv \argmax_{c\in[0,1]} g_n(c)\;.
\end{equation} 
In~\cref{app:auxfunprops} we prove some properties of $g_n$ that will be used in the rest of this section.

\subsection{Upper Bounds}
\label{sec:upper} This section is dedicated to proving the main result of our paper.
First (\cref{th:main_upper}) we show that pricing is indeed asymptotically optimal
with respect to revenue and then (\cref{th:main_upper_small}) we also provide a more
refined upper-bound on the approximation ratio that is fine-tuned to work well for a
small number of bidders $n$. As we will see in the following \cref{sec:lower}, our
upper bound analysis of this section is essentially tight (see also
\cref{fig:approx_lower}).

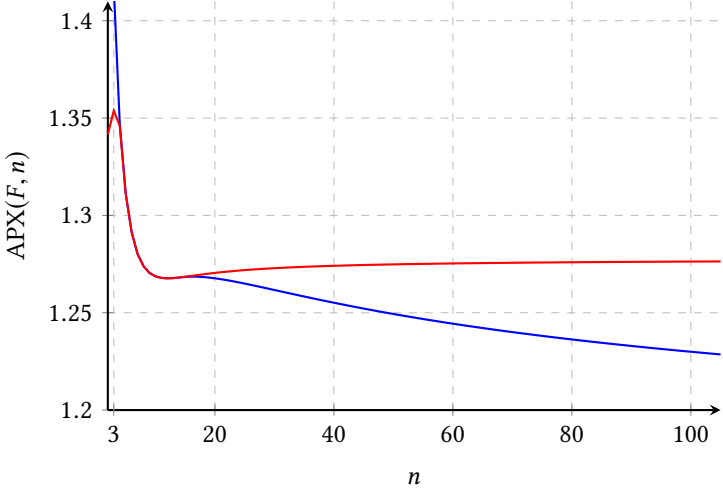
\begin{figure}
  \centering
\begin{tikzpicture}
\begin{axis}[
  width=0.7\textwidth,
    height=0.5\textwidth,
  xmin=2, xmax=105,
  ymin=1.2, ymax=1.41,
    axis lines = left,
    ytick={1.20,1.25,...,1.45},
    ymajorgrids=true,
    xmajorgrids=true,
    extra x ticks={3},
    grid style=dashed,
    xlabel = $n$,
    ylabel = {$\apx{F}{n}$},
    thick
]
\addplot[
    color=blue,
    ]
    coordinates {(2,1.57048)(3,1.4173)(4,1.34746)(5,1.31129)(6,1.29139)(7,1.28019)(8,1.2739)(9,1.27045)(10,1.26868)(11,1.2679)(12,1.26769)(13,1.2678)(14,1.2681)(15,1.26842)(16,1.26856)(17,1.26853)(18,1.26836)(19,1.26808)(20,1.2677)(21,1.26725)(22,1.26674)(23,1.26619)(24,1.2656)(25,1.26498)(26,1.26434)(27,1.26369)(28,1.26303)(29,1.26236)(30,1.26168)(31,1.26101)(32,1.26034)(33,1.25967)(34,1.259)(35,1.25834)(36,1.25769)(37,1.25704)(38,1.2564)(39,1.25577)(40,1.25515)(41,1.25453)(42,1.25392)(43,1.25332)(44,1.25273)(45,1.25215)(46,1.25158)(47,1.25101)(48,1.25045)(49,1.2499)(50,1.24936)(51,1.24883)(52,1.24831)(53,1.24779)(54,1.24728)(55,1.24678)(56,1.24628)(57,1.2458)(58,1.24532)(59,1.24484)(60,1.24438)(61,1.24392)(62,1.24346)(63,1.24302)(64,1.24258)(65,1.24214)(66,1.24171)(67,1.24129)(68,1.24087)(69,1.24046)(70,1.24006)(71,1.23966)(72,1.23926)(73,1.23887)(74,1.23849)(75,1.23811)(76,1.23773)(77,1.23736)(78,1.237)(79,1.23664)(80,1.23628)(81,1.23593)(82,1.23558)(83,1.23524)(84,1.2349)(85,1.23456)(86,1.23423)(87,1.2339)(88,1.23358)(89,1.23326)(90,1.23294)(91,1.23263)(92,1.23232)(93,1.23201)(94,1.23171)(95,1.23141)(96,1.23112)(97,1.23082)(98,1.23053)(99,1.23025)(100,1.22996)(101,1.22968)(102,1.2294)(103,1.22913)(104,1.22886)(105,1.22859)
    };
\addplot[
    color=red,
    ]
    coordinates {(2,1.34159)(3,1.35384)(4,1.34601)(5.,1.31129)(6.,1.29139)(7.,1.28019)(8.,1.2739)(9.,1.27045)(10.,1.26868)(11.,1.2679)(12.,1.26769)(13.,1.2678)(14.,1.2681)(15.,1.26848)(16.,1.2689)(17.,1.26931)(18.,1.26972)(19.,1.27011)(20.,1.27047)(21.,1.2708)(22.,1.27111)(23.,1.2714)(24.,1.27166)(25.,1.27191)(26.,1.27213)(27.,1.27234)(28.,1.27254)(29.,1.27272)(30.,1.27289)(31.,1.27305)(32.,1.2732)(33.,1.27334)(34.,1.27347)(35.,1.27359)(36.,1.27371)(37.,1.27382)(38.,1.27392)(39.,1.27402)(40.,1.27411)(41.,1.2742)(42.,1.27429)(43.,1.27437)(44.,1.27445)(45.,1.27452)(46.,1.27459)(47.,1.27466)(48.,1.27472)(49.,1.27479)(50.,1.27484)(51.,1.2749)(52.,1.27496)(53.,1.27501)(54.,1.27506)(55.,1.27511)(56.,1.27516)(57.,1.2752)(58.,1.27525)(59.,1.27529)(60.,1.27533)(61.,1.27537)(62.,1.27541)(63.,1.27544)(64.,1.27548)(65.,1.27551)(66.,1.27555)(67.,1.27558)(68.,1.27561)(69.,1.27564)(70.,1.27567)(71.,1.2757)(72.,1.27573)(73.,1.27576)(74.,1.27579)(75.,1.27581)(76.,1.27584)(77.,1.27586)(78.,1.27589)(79.,1.27591)(80.,1.27593)(81.,1.27595)(82.,1.27598)(83.,1.276)(84.,1.27602)(85.,1.27604)(86.,1.27606)(87.,1.27608)(88.,1.2761)(89.,1.27611)(90.,1.27613)(91.,1.27615)(92.,1.27617)(93.,1.27618)(94.,1.2762)(95.,1.27622)(96.,1.27623)(97.,1.27625)(98.,1.27626)(99.,1.27628)(100.,1.27629)(101.,1.27631)(102.,1.27632)(103.,1.27633)(104.,1.27635)(105.,1.27636)
    };
\end{axis}
\end{tikzpicture}
\Description{The upper bounds on the approximation ratio of anonymous pricing for
$n$ i.i.d.\ bidders with MHR valuations, given by~\cref{th:main_upper} and
\cref{th:main_upper_small}. The best (smallest) of the two converges to the optimal
value of $1$ as the number of bidders grows large, at a rate of $1+O\left(\ln \ln
n/\ln n\right)$. A single, unified plot of this can be seen in
\cref{fig:approx_lower}, together with a matching lower bound. In the worst case
($n=3$), our upper bound is at most $1.354$.}
  \caption{The upper bounds on the approximation ratio $\apx{F}{n}$ of anonymous
  pricing for $n$ i.i.d.\ bidders with MHR valuations, given by~\cref{th:main_upper}
  (blue) and \cref{th:main_upper_small} (red). The best (smallest) of the two
  converges to the optimal value of $1$ as the number of bidders grows large, at a
  rate of $1+O\left(\ln \ln n/\ln n\right)$. A single, unified plot of this can be
  seen in \cref{fig:approx_lower} (black), together with a matching lower bound
  (red). In the worst case ($n=3$), our upper bound is at most $1.354$.}
\label{fig:approx_upper}
\end{figure}

\begin{theorem}
\label{th:main_upper}
Using the same take-it-or-leave-it price, to sell an item to $n$ buyers with i.i.d.\
valuations from a continuous MHR distribution $F$, is asymptotically optimal with
respect to revenue. In particular,
$$
\apx{F}{n}=1+O\left(\frac{\ln \ln n}{\ln
n}\right)\;.
$$
\end{theorem}
A plot of the exact values of this upper bound (given by~\eqref{eq:apx_bound_1} below) can be seen in~\cref{fig:approx_upper} (blue).
\begin{proof}
First notice that by using the monopoly reserve price $r^*$ of $F$ as a
take-it-or-leave it price to the $n$ bidders, we get an expected revenue of
\begin{equation}
\label{eq:price1_bound}
\pricep{F}{n}{r^*}
=r^*(1-F(r^*)^n)
=r^*[1-(1-q^*)^n]
= R(q^*)\frac{1-(1-q^*)^n}{q^*}\;,
\end{equation}
where $q^*=1-F(r^*)$ is the quantile of the monopoly reserve price, for which we know that $q^*\geq\frac{1}{e}$ (\cref{lem:bounds_reserve_price}), and $R$ denotes the revenue curve (see~\cref{sec:model}).

Next, for simplicity denote $\nu = \expect{\orderstat{X}{n-1}{n}}$.
For any real $c\in [0,1]$, if we offer a price of $c\cdot \nu$ we have
\begin{equation}
\label{eq:bound_price_second_order_c}
\pricep{F}{n}{c\nu}
=c\nu[1-F(c\nu)^n]
\geq c\nu\left[1-\left(1-e^{-c(H_n-1)}\right)^n\right]\;,
\end{equation}
the inequality holding due to \cref{lemma:bound_tail_orderstats} (for $k=n-1$). Optimizing with respect to $c$ we get that
\begin{equation}
\label{eq:bound_price_second_order_opt}
\price{F}{n} \geq \nu \max_{c\in[0,1]}g_n(c)\;.
\end{equation}

Using the two lower bounds \eqref{eq:price1_bound} and
\eqref{eq:bound_price_second_order_opt} on the pricing revenue, in conjunction with
the upper bound on the optimal revenue from \cref{lemma:opt_bound_Fu_et_al} we can
bound the approximation ratio of pricing by
\begin{align}
\apx{F}{n}&= \frac{\myerson{F}{n}}{\price{F}{n}} \notag\\
  &\leq \frac{\nu}{\nu \max_{c\in[0,1]}g_n(c)}+ \frac{R(q^*)(1-q^*)^{n-1}}{R(q^*)\frac{1-(1-q^*)^n}{q^*}}\notag\\
  &= \frac{1}{\max_{c\in[0,1]}g_n(c)} + \frac{q^*(1-q^*)^{n-1}}{1-(1-q^*)^n}\label{eq:apx_bound_3}\\
  & \leq \frac{1}{\max_{c\in[0,1]}g_n(c)} + \frac{(e-1)^{n-1}}{e^n-(e-1)^{n}} \label{eq:apx_bound_1}\\
  &= 1+O\left(\frac{\ln\ln n}{\ln n}\right)+O\left(\left(\frac{e}{e-1}\right)^{-n}\right)\;. \label{eq:apx_bound_2}
\end{align}
\Cref{eq:apx_bound_1} holds by observing that function
$x\mapsto\frac{x(1-x)^{n-1}}{1-(1-x)^n}$ is decreasing over $(0,1]$, for any $n\geq
2$, and taking into consideration that $q^*\geq 1/e$, while for
\eqref{eq:apx_bound_2} we make use of the asymptotics from  \cref{lem:3}. The upper
bound given by \eqref{eq:apx_bound_1} is plotted by the blue line
in~\cref{fig:approx_upper}.
\end{proof}

\begin{theorem}
\label{th:main_upper_small} The approximation ratio of the revenue obtained by using
the same take-it-or-leave-it price, to sell an item to $n$ buyers with i.i.d.\
valuations from a continuous MHR distribution $F$, is at most
$$
\apx{F}{n}\leq
\max_{q\in[1/e,1]}\min\sset{\frac{1}{1-(1-e^{-H_n+1})^n}+\frac{q(1-q)^{n-1}}{1-(1-q)^n},\frac{\int_0^q\frac{1-(1-z)^n}{z}\,dz}{1-(1-q)^n}}\;.
$$
In particular, the worst case (maximum) of this quantity is attained at $n=3$ and is
at most $\apx{F}{3}\leq 1.354$.
\end{theorem}
A plot of the exact values of this upper bound can
be seen in~\cref{fig:approx_upper} (red).
\begin{proof}
From \eqref{eq:apx_bound_3} in the proof of \cref{th:main_upper} we
can get the following upper bound on the approximation ratio, by using (possibly
suboptimally) $c\gets 1$ for the maximization operator:
$$
\apx{F}{n}\leq \frac{1}{g_n(1)} + \frac{q^*(1-q^*)^{n-1}}{1-(1-q^*)^n}
=\frac{1}{1-(1-e^{-H_n+1})^n}+\frac{q^*(1-q^*)^{n-1}}{1-(1-q^*)^n}\;.
$$
On the other hand, using the reserve price of $F$ as a price and combining the
guarantee of \eqref{eq:price1_bound} with the upper bound on the optimal revenue
from \cref{lem:opt_bound_2}, gives us
$$
\apx{F}{n} \leq \frac{r^*\int_{0}^{q^*}\frac{1-(1-z)^{n}}{z}\,
dz}{R(q^*)\frac{1-(1-q^*)^n}{q^*}}=
\frac{\int_0^{q^*}\frac{1-(1-z)^n}{z}\,dz}{1-(1-q^*)^n}\;,
$$
since $R(q^*)=r^*q^*$. Recalling that $q^*\in[1/e,1]$ and taking the best (i.e.,
minimum) of the two bounds above, finishes the proof.
\end{proof}

\subsection{Lower Bound}
\label{sec:lower} The lower bound instance of this section (\cref{th:lower_expo})
shows that our main positive result for the approximation ratio of pricing under MHR
distributions in \cref{th:main_upper} is essentially tight (see also
\cref{fig:approx_lower}). It is achieved by an exponential distribution instance.
Before proving it, we need the following auxiliary lemma about the maximizers of
functions $g_n$ that we introduced in \eqref{eq:gdef}. Its proof can be found in
\cref{app:auxfunprops}.

\begin{lemma}
\label{lem:propsg}
For any positive integer $n$, function $g_n$ (defined in
\eqref{eq:gdef}) has a unique maximizer. Furthermore, for all $n\geq 17$,
$$
\argmax_{c\geq 0} g_n(c) \leq 1\;.
$$
\end{lemma}

\begin{theorem}
\label{th:lower_expo}
For $n\geq 2$ bidders with exponentially i.i.d.\ valuations,
the approximation ratio of anonymous pricing is at least
$$ \apx{\mathcal{E}}{n}\geq\frac{1}{\max_{c\geq 0}g_n(c)}\;,$$
where function $g_n$ is defined in
\eqref{eq:gdef} and $\mathcal E$ is the exponential distribution. In particular, the upper
bound derived in the proof of \cref{th:main_upper} is tight (up to an exponentially
vanishing \emph{additive} factor).
\end{theorem}
A plot of the
lower bound given by the theorem above can be seen in \cref{fig:approx_lower} (red).
\begin{proof}
Let $X\sim \mathcal E$ be an exponential random variable. 
Since the revenue of a second-price auction cannot be greater than the optimal one, we have 
$$
\myerson{\mathcal{E}}{n}
\geq \expect{\orderstat{X}{n-1}{n}}
=H_n-1,
$$
where the equality is taken from~\cref{append:prob_basics}. Furthermore,
$$
\price{\mathcal E}{n}
= \sup_{x\geq 0}x\left(1-F_{\mathcal E}^n(x)\right)
= \sup_{x\geq 0}x\left[1-\left(1-e^{-x}\right)^n\right]
= (H_n-1)\max_{c\geq 0}g_n(c)\;.
$$
Putting the above together, we get the desired lower bound on the approximation
ratio.

For the tightness, we need to show that our lower bound is within an additive,
exponentially decreasing factor of the upper bound given in \eqref{eq:apx_bound_1}.
Since the second term in \eqref{eq:apx_bound_1} is at most
$O\left(\left(\frac{e}{e-1}\right)^{-n}\right)$, it is enough to show that, for a
sufficiently large number of bidders $n$,
$$
\max_{c\in[0,\infty)} g_n(c)= \max_{c\in[0,1]} g_n(c)\;.
$$
This is exactly what we proved in~\cref{lem:propsg}, for any $n\geq 17$.
\end{proof}

\subsection{Explicit Prices -- Knowledge of the Distribution}
\label{sec:prior_indi}
\begin{figure}
  \centering
\begin{tikzpicture}
\begin{axis}[
  width=0.7\textwidth,
    height=0.5\textwidth,
  xmin=3, xmax=32,
  ymin=1.00, ymax=1.75,
    axis lines = left,
    ytick={1.00,1.1,...,1.7},
    ymajorgrids=true,
    xmajorgrids=true,
    extra y ticks={1.03974,1.70722},
    extra x ticks={3},
    grid style=dashed,
    xlabel = $n$,
    ylabel = {Approximation ratio},
    thick
]
\addplot[
    color=black,
    ]
    coordinates {(2,1.34159)(3,1.35384)(4,1.34601)(5,1.31129)(6,1.29139)(7,1.28019)(8,1.2739)(9,1.27045)(10,1.26868)(11,1.2679)(12,1.26769)(13,1.2678)(14,1.2681)(15,1.26842)(16,1.26856)(17,1.26853)(18,1.26836)(19,1.26808)(20,1.2677)(21,1.26725)(22,1.26674)(23,1.26619)(24,1.2656)(25,1.26498)(26,1.26434)(27,1.26369)(28,1.26303)(29,1.26236)(30,1.26168)(31,1.26101)(32,1.26034)(33,1.25967)(34,1.259)(35,1.25834)(36,1.25769)(37,1.25704)(38,1.2564)(39,1.25577)(40,1.25515)(41,1.25453)(42,1.25392)(43,1.25332)(44,1.25273)(45,1.25215)(46,1.25158)(47,1.25101)(48,1.25045)(49,1.2499)(50,1.24936)(51,1.24883)(52,1.24831)(53,1.24779)(54,1.24728)(55,1.24678)(56,1.24628)(57,1.2458)(58,1.24532)(59,1.24484)(60,1.24438)(61,1.24392)(62,1.24346)(63,1.24302)(64,1.24258)(65,1.24214)(66,1.24171)(67,1.24129)(68,1.24087)(69,1.24046)(70,1.24006)(71,1.23966)(72,1.23926)(73,1.23887)(74,1.23849)(75,1.23811)(76,1.23773)(77,1.23736)(78,1.237)(79,1.23664)(80,1.23628)(81,1.23593)(82,1.23558)(83,1.23524)(84,1.2349)(85,1.23456)(86,1.23423)(87,1.2339)(88,1.23358)(89,1.23326)(90,1.23294)(91,1.23263)(92,1.23232)(93,1.23201)(94,1.23171)(95,1.23141)(96,1.23112)(97,1.23082)(98,1.23053)(99,1.23025)(100,1.22996)(101,1.22968)(102,1.2294)(103,1.22913)(104,1.22886)(105,1.22859)
    };
    \addlegendentry{Optimal pricing -- upper bound}
\addplot[
    color=red,
    ]
    coordinates {(2,1)(3,1.03974)(4,1.13595)(5,1.1859)(6,1.21487)(7,1.23291)(8,1.24467)(9,1.25256)(10,1.25797)(11,1.2617)(12,1.26426)(13,1.266)(14,1.26713)(15,1.26782)(16,1.26818)(17,1.26829)(18,1.26821)(19,1.26798)(20,1.26764)(21,1.26721)(22,1.26672)(23,1.26617)(24,1.26559)(25,1.26497)(26,1.26434)(27,1.26369)(28,1.26303)(29,1.26236)(30,1.26168)(31,1.26101)(32,1.26034)(33,1.25967)(34,1.259)(35,1.25834)(36,1.25769)(37,1.25704)(38,1.2564)(39,1.25577)(40,1.25515)(41,1.25453)(42,1.25392)(43,1.25332)(44,1.25273)(45,1.25215)(46,1.25158)(47,1.25101)(48,1.25045)(49,1.2499)(50,1.24936)(51,1.24883)(52,1.24831)(53,1.24779)(54,1.24728)(55,1.24678)(56,1.24628)(57,1.2458)(58,1.24532)(59,1.24484)(60,1.24438)(61,1.24392)(62,1.24346)(63,1.24302)(64,1.24258)(65,1.24214)(66,1.24171)(67,1.24129)(68,1.24087)(69,1.24046)(70,1.24006)(71,1.23966)(72,1.23926)(73,1.23887)(74,1.23849)(75,1.23811)(76,1.23773)(77,1.23736)(78,1.237)(79,1.23664)(80,1.23628)(81,1.23593)(82,1.23558)(83,1.23524)(84,1.2349)(85,1.23456)(86,1.23423)(87,1.2339)(88,1.23358)(89,1.23326)(90,1.23294)(91,1.23263)(92,1.23232)(93,1.23201)(94,1.23171)(95,1.23141)(96,1.23112)(97,1.23082)(98,1.23053)(99,1.23025)(100,1.22996)(101,1.22968)(102,1.2294)(103,1.22913)(104,1.22886)(105,1.22859)
    };
    \addlegendentry{Optimal pricing -- lower bound}
\addplot[
  color=blue
  ]
  coordinates {(2,2.05371)(3,1.70722)(4,1.52147)(5,1.41642)(6,1.35529)(7,1.31921)(8,1.29782)(9,1.28517)(10,1.27776)(11,1.27351)(12,1.27117)(13,1.26997)(14,1.26945)(15,1.26926)(16,1.26909)(17,1.26886)(18,1.26857)(19,1.26821)(20,1.26778)(21,1.2673)(22,1.26678)(23,1.26621)(24,1.26561)(25,1.26499)(26,1.26435)(27,1.26369)(28,1.26303)(29,1.26236)(30,1.26169)(31,1.26101)(32,1.26034)(33,1.25967)(34,1.259)(35,1.25834)(36,1.25769)(37,1.25704)(38,1.2564)(39,1.25577)(40,1.25515)(41,1.25453)(42,1.25392)(43,1.25332)(44,1.25273)(45,1.25215)(46,1.25158)(47,1.25101)(48,1.25045)(49,1.2499)(50,1.24936)(51,1.24883)(52,1.24831)(53,1.24779)(54,1.24728)(55,1.24678)(56,1.24628)(57,1.2458)(58,1.24532)(59,1.24484)(60,1.24438)(61,1.24392)(62,1.24346)(63,1.24302)(64,1.24258)(65,1.24214)(66,1.24171)(67,1.24129)(68,1.24087)(69,1.24046)(70,1.24006)(71,1.23966)(72,1.23926)(73,1.23887)(74,1.23849)(75,1.23811)(76,1.23773)(77,1.23736)(78,1.237)(79,1.23664)(80,1.23628)(81,1.23593)(82,1.23558)(83,1.23524)(84,1.2349)(85,1.23456)(86,1.23423)(87,1.2339)(88,1.23358)(89,1.23326)(90,1.23294)(91,1.23263)(92,1.23232)(93,1.23201)(94,1.23171)(95,1.23141)(96,1.23112)(97,1.23082)(98,1.23053)(99,1.23025)(100,1.22996)(101,1.22968)(102,1.2294)(103,1.22913)(104,1.22886)(105,1.22859)};
  \addlegendentry{Pricing at $c_n\expectsmall{\orderstat{X}{n-1}{n}}$}
  \addplot[
    color=black,
    ]
    coordinates {(2,1.34159)(3,1.35384)(4,1.34601)(5,1.31129)(6,1.29139)(7,1.28019)(8,1.2739)(9,1.27045)(10,1.26868)(11,1.2679)(12,1.26769)(13,1.2678)(14,1.2681)(15,1.26842)(16,1.26856)(17,1.26853)(18,1.26836)(19,1.26808)(20,1.2677)(21,1.26725)(22,1.26674)(23,1.26619)(24,1.2656)(25,1.26498)(26,1.26434)(27,1.26369)(28,1.26303)(29,1.26236)(30,1.26168)(31,1.26101)(32,1.26034)(33,1.25967)(34,1.259)(35,1.25834)(36,1.25769)(37,1.25704)(38,1.2564)(39,1.25577)(40,1.25515)(41,1.25453)(42,1.25392)(43,1.25332)(44,1.25273)(45,1.25215)(46,1.25158)(47,1.25101)(48,1.25045)(49,1.2499)(50,1.24936)(51,1.24883)(52,1.24831)(53,1.24779)(54,1.24728)(55,1.24678)(56,1.24628)(57,1.2458)(58,1.24532)(59,1.24484)(60,1.24438)(61,1.24392)(62,1.24346)(63,1.24302)(64,1.24258)(65,1.24214)(66,1.24171)(67,1.24129)(68,1.24087)(69,1.24046)(70,1.24006)(71,1.23966)(72,1.23926)(73,1.23887)(74,1.23849)(75,1.23811)(76,1.23773)(77,1.23736)(78,1.237)(79,1.23664)(80,1.23628)(81,1.23593)(82,1.23558)(83,1.23524)(84,1.2349)(85,1.23456)(86,1.23423)(87,1.2339)(88,1.23358)(89,1.23326)(90,1.23294)(91,1.23263)(92,1.23232)(93,1.23201)(94,1.23171)(95,1.23141)(96,1.23112)(97,1.23082)(98,1.23053)(99,1.23025)(100,1.22996)(101,1.22968)(102,1.2294)(103,1.22913)(104,1.22886)(105,1.22859)
    };
\end{axis}
\end{tikzpicture}
\Description{Bounds on the approximation ratio of anonymous pricing for
  $n=3$ up to $n=30$ i.i.d. bidders with MHR valuations: the upper bound on optimal
  pricing derived in \cref{sec:upper} (see also \cref{fig:approx_upper}),
  the lower bound given by \cref{th:lower_expo}, and the upper bound of
  pricing at the expected value of the second-highest order statistic, scaled down
  by parameter $c_n$, given in \cref{cor:upper_bound_prior_indi} of
  \cref{sec:prior_indi}. They are all (asymptotically) optimal, their (additive)
  difference decreasing exponentially fast. They all converge to the optimal value
  of $1$ at a rate of $1+O(\ln\ln n/\ln n)$.}
  \caption{Bounds on the approximation ratio of anonymous pricing for
  $n=3,\dots,30$ i.i.d.\ bidders with MHR valuations: the upper bound on optimal
  pricing (black) derived in \cref{sec:upper} (see also \cref{fig:approx_upper}),
  the lower bound (red) given by \cref{th:lower_expo}, and the upper bound of
  pricing at the expected value of the second-highest order statistic, scaled down
  by parameter $c_n$ (blue), given in \cref{cor:upper_bound_prior_indi} of
  \cref{sec:prior_indi}. They are all (asymptotically) optimal, their (additive)
  difference decreasing exponentially fast. They all converge to the optimal value
  of $1$ at a rate of $1+O(\ln\ln n/\ln n)$.}
\label{fig:approx_lower}
\end{figure}
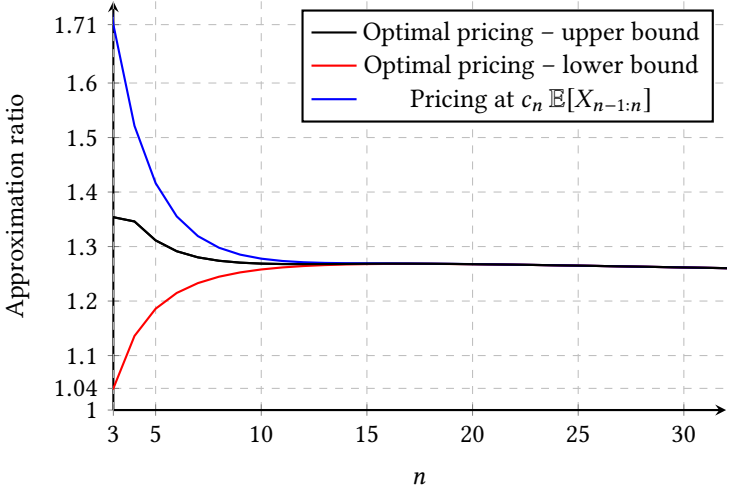

Our main result from \cref{sec:upper} demonstrates that a seller, facing $n$ bidders
with i.i.d.\ valuations from an MHR distribution $F$, can achieve (asymptotically)
optimal revenue by using just an anonymous, take-it-or-leave-it price. Taking a
careful look within the proof of \cref{th:main_upper}, we see that this upper bound
is derived by comparing the optimal Myersonian revenue (via the bound provided by
\cref{lemma:opt_bound_Fu_et_al}) to that of two different anonymous pricings;
namely, first (see \eqref{eq:price1_bound}) we use the monopoly reserve $r^*$ of $F$,
and then (see \eqref{eq:bound_price_second_order_c} and \eqref{eq:bound_price_second_order_opt})
a multiple of the expectation $\nu=\expectsmall{\orderstat{X}{n-1}{n}}$ of the
second-highest order statistic of $F$, in particular $c_n\cdot \nu$ where $c_n=\argmax_{c\in[0,1]} g_n(c)$ was defined in~\eqref{eq:cparsdef}.
Although the latter price requires only the knowledge of
$\nu=\expectsmall{\orderstat{X}{n-1}{n}}$, that is not the case for the former;
determining the reserve price $r^*$ demands, in general, a detailed knowledge of the
distribution $F$: it is the maximizer of $r(1-F(r))$.

As a result, we would ideally like to provide a more robust solution, that would
still provide optimality but depend only on limited information about $F$. If we pay
even closer attention to the proof of \cref{th:main_upper}, and the derivation of
\eqref{eq:apx_bound_2} in particular, we will see that the summand of our upper
bound that corresponds to the pricing using $r^*$ is exponentially decreasing,
according to $\left(\frac{e}{e-1}\right)^{-n}$. Therefore, if we could show that the
expected revenue achieved by using an anonymous price of $c_n \nu$ is within a constant
factor from that of using an anonymous price of $r^*$, then we could deduce that using only price $c_n
\nu$ yields essentially the same approximation ratio (and in particular, asymptotically optimal revenue). We now proceed to formalize this line of reasoning.

\begin{lemma}
\label{lem:reserve_ordstat_comparison} For $n\geq 2$ bidders with i.i.d.\ valuations
from an MHR distribution $F$ with monopoly reserve $r^*$ and parameters
$c_n\in[0,1]$ given by \eqref{eq:cparsdef},
$$
\pricep{F}{n}{c_n\cdot \expectsmall{\orderstat{X}{n-1}{n}}}\geq
(1-o(1))\frac{e-1}{e} \cdot \pricep{F}{n}{r^*}\;,
$$
where $X\sim F$.
\end{lemma}
\begin{proof}
For convenience, denote $\mu=\expectsmall{X}$ and $\nu =
\expectsmall{\orderstat{X}{n-1}{n}}$. By the proof of \cref{th:main_upper} (see
\eqref{eq:bound_price_second_order_c} and \eqref{eq:bound_price_second_order_opt}) we know that
by offering an anonymous price of $c_n\cdot \nu$ gives us an expected revenue of at
least
\begin{equation*}
\pricep{F}{n}{c_n\cdot \nu} \geq \nu \max_{c\in[0,1]}g_n(c)
    \geq\max_{c\in[0,1]}g_n(c)\frac{n-H_n}{n-1}\cdot \mu\;,
\end{equation*}
the second inequality holding due to~\cref{lem:bound_2ndordstat_mean_MHR}.

On the other hand, from \eqref{eq:price1_bound} we know that using the reserve price
$r^*$ as an anonymous price to all bidders gives an expected revenue of at most
\begin{equation*}
\pricep{F}{n}{r^*} = r^*[1-(1-q^*)^n]
    \leq \frac{\ln(1/q^*)}{1-q^*}[1-(1-q^*)^n] \cdot \mu\;,
\end{equation*}
the inequality holding due to \cref{lem:bounds_reserve_price}.

Putting everything together, we finally get that
\begin{align}
\frac{\pricep{F}{n}{r^*}}{\pricep{F}{n}{c_n\nu}}
&\leq
\frac{\ln(1/q^*)}{1-q^*}[1-(1-q^*)^n]\frac{n-1}{n-H_n}\frac{1}{\max_{c\in[0,1]}g_n(c)}\label{eq:helper7}\\
&\leq \frac{e}{e-1}\frac{n-1}{n-H_n}\frac{1}{\max_{c\in[0,1]}g_n(c)}\notag\\
&\leq (1+o(1))\frac{e}{e-1}\;.\notag
\end{align}
The second inequality holds because $\frac{\ln(1/q^*)}{1-q^*}[1-(1-q^*)^n]\leq
\frac{\ln(1/q^*)}{1-q^*}\leq \frac{e}{e-1}$, since the function
$x\mapsto\frac{\ln(1/x)}{1-x}$ is decreasing for $x>0$ and $q^*\geq 1/e$ (from
Property~1 of \cref{lem:bounds_reserve_price}). The last inequality is a consequence
of \cref{lem:3} and the fact that $H_n\leq \ln(n)+1$.
\end{proof}

As discussed before, \cref{lem:reserve_ordstat_comparison} shows us that there
indeed exists an anonymous price that depends on the knowledge of only the
expectation of the second-highest order statistic of the valuation distribution and which,
furthermore, guarantees an (asymptotically) optimal revenue. We can even provide a
closed-form upper bound for it:
\begin{theorem}
\label{cor:upper_bound_prior_indi} Let $F$ be an MHR distribution and
$\orderstat{X}{n-1}{n}$ denote the second-highest out of $n$ i.i.d.\ draws from
$F$. Then, using an anonymous price of $c_n\cdot
\expectsmall{\orderstat{X}{n-1}{n}}$, where $c_n$ is given in \eqref{eq:cparsdef}, to
sell an item to $n\geq 2$ bidders with i.i.d.\ valuations from $F$, guarantees a
revenue with approximation ratio of at most
$$
\frac{\myerson{F}{n}}{\pricep{F}{n}{c_n\expectsmall{\orderstat{X}{n-1}{n}}}}
\leq
\frac{1}{\max_{c\in[0,1]}g_n(c)}\left[1+
\frac{1}{e}\frac{n-1}{n-H_n}\left(\frac{e-1}{e}\right)^{n-2}\right]\;.
$$
\end{theorem}
A plot of this upper bound can be seen in~\cref{fig:approx_lower} (blue).
\begin{proof} Simulating the proof of the approximation upper bound in
\cref{th:main_upper}, but now using \eqref{eq:helper7} to approximate
$\pricep{F}{n}{r^*}$ by $\pricep{F}{n}{c_n\expectsmall{\orderstat{X}{n-1}{n}}}$, the
derivation in \eqref{eq:apx_bound_3} gives us that
\begin{align*}
\frac{\myerson{F}{n}}{\pricep{F}{n}{c_n\expectsmall{\orderstat{X}{n-1}{n}}}}
&\leq \frac{1}{\max_{c\in[0,1]}g_n(c)}\\
&+ \frac{\ln(1/q^*)}{1-q^*}[1-(1-q^*)^n]\frac{n-1}{n-H_n}\frac{1}{\max_{c\in[0,1]}g_n(c)}\cdot \frac{q^*(1-q^*)^{n-1}}{1-(1-q^*)^n}\\
&=\frac{1}{\max_{c\in[0,1]}g_n(c)}\left[1+ \frac{n-1}{n-H_n}\ln(1/q^*)q^*(1-q^*)^{n-2}\right]\\
&\leq \frac{1}{\max_{c\in[0,1]}g_n(c)}\left[1+ \frac{1}{e}\frac{n-1}{n-H_n}\left(\frac{e-1}{e}\right)^{n-2}\right]\;,
\end{align*}
the last inequality coming from \cref{lem:helper_bound_1}, together with the fact
that $q^*\in[1/e,1]$ (see Property~1 of \cref{lem:bounds_reserve_price}).
\end{proof}

\section{Bounds for \texorpdfstring{$\lambda$}{lambda}-Regular Distributions}
\label{sec:lambdareg}
\begin{figure}
\centering
\scalebox{0.55}{
\begin{tikzpicture}
\begin{axis}[
  xmin=0, xmax=1,
  ymin=1, ymax=1.7,
    ytick={1,1.1,...,1.8},
    ymajorgrids=true,
    xmajorgrids=true,
    grid style=dashed,
    xlabel = {$\lambda$},
    ylabel = {$\apx{5}{\lambda}$},
    thick
]
\addplot[
    color=red,
    ]
    coordinates {(0.001,1.00118) (0.01,1.01094) (0.02,1.02112) (0.03,1.03091) (0.04,1.04041) (0.05,1.04967) (0.06,1.05874) (0.07,1.06764) (0.08,1.07638) (0.09,1.08497) (0.1,1.09342) (0.11,1.10174) (0.12,1.10994) (0.13,1.11801) (0.14,1.12597) (0.15,1.13381) (0.16,1.14154) (0.17,1.14916) (0.18,1.15666) (0.19,1.16405) (0.2,1.17133) (0.21,1.17851) (0.22,1.18557) (0.23,1.19252) (0.24,1.19935) (0.25,1.20607) (0.26,1.21268) (0.27,1.21918) (0.28,1.22555) (0.29,1.23181) (0.3,1.23794) (0.31,1.24396) (0.32,1.24984) (0.33,1.25561) (0.34,1.26124) (0.35,1.26674) (0.36,1.2721) (0.37,1.27733) (0.38,1.28242) (0.39,1.28737) (0.4,1.29217) (0.41,1.29683) (0.42,1.30133) (0.43,1.30568) (0.44,1.30986) (0.45,1.31389) (0.46,1.31775) (0.47,1.32144) (0.48,1.32495) (0.49,1.32829) (0.5,1.33144) (0.51,1.3344) (0.52,1.33717) (0.53,1.33975) (0.54,1.34212) (0.55,1.34447) (0.56,1.34683) (0.57,1.34923) (0.58,1.35165) (0.59,1.3541) (0.6,1.35658) (0.61,1.35909) (0.62,1.36163) (0.63,1.36421) (0.64,1.36681) (0.65,1.36944) (0.66,1.37211) (0.67,1.3748) (0.68,1.37753) (0.69,1.3803) (0.7,1.3831) (0.71,1.38594) (0.72,1.38881) (0.73,1.39172) (0.74,1.39466) (0.75,1.39765) (0.76,1.40067) (0.77,1.40373) (0.78,1.40684) (0.79,1.40998) (0.8,1.41317) (0.81,1.4164) (0.82,1.41968) (0.83,1.423) (0.84,1.42636) (0.85,1.42978) (0.86,1.43324) (0.87,1.43675) (0.88,1.44031) (0.89,1.44392) (0.9,1.44758) (0.91,1.4513) (0.92,1.45507) (0.93,1.4589) (0.94,1.46279) (0.95,1.46673) (0.96,1.47074) (0.97,1.4748) (0.98,1.47893) (0.99,1.48313) (1.,1.48739)
    };
\addplot[
    color=blue,
    ]
    coordinates {(0.001,1.31148) (0.01,1.3132) (0.02,1.31513) (0.03,1.31708) (0.04,1.31904) (0.05,1.32103) (0.06,1.32304) (0.07,1.32507) (0.08,1.32712) (0.09,1.32919) (0.1,1.33128) (0.11,1.33339) (0.12,1.33552) (0.13,1.33768) (0.14,1.33985) (0.15,1.34205) (0.16,1.34428) (0.17,1.34652) (0.18,1.34879) (0.19,1.35109) (0.2,1.35341) (0.21,1.35575) (0.22,1.35812) (0.23,1.36052) (0.24,1.36294) (0.25,1.36539) (0.26,1.36786) (0.27,1.37036) (0.28,1.3729) (0.29,1.37545) (0.3,1.37804) (0.31,1.38066) (0.32,1.38331) (0.33,1.38598) (0.34,1.38869) (0.35,1.39143) (0.36,1.39421) (0.37,1.39701) (0.38,1.39985) (0.39,1.40272) (0.4,1.40563) (0.41,1.40857) (0.42,1.41155) (0.43,1.41456) (0.44,1.41761) (0.45,1.4207) (0.46,1.42382) (0.47,1.42699) (0.48,1.4302) (0.49,1.43344) (0.5,1.43673) (0.51,1.44006) (0.52,1.44344) (0.53,1.44686) (0.54,1.45032) (0.55,1.45383) (0.56,1.45739) (0.57,1.461) (0.58,1.46465) (0.59,1.46836) (0.6,1.47211) (0.61,1.47592) (0.62,1.47978) (0.63,1.4837) (0.64,1.48768) (0.65,1.49171) (0.66,1.4958) (0.67,1.49995) (0.68,1.50416) (0.69,1.50844) (0.7,1.51278) (0.71,1.51719) (0.72,1.52167) (0.73,1.52621) (0.74,1.53083) (0.75,1.53553) (0.76,1.5403) (0.77,1.54515) (0.78,1.55008) (0.79,1.5551) (0.8,1.5602) (0.81,1.56539) (0.82,1.57068) (0.83,1.57605) (0.84,1.58153) (0.85,1.58711) (0.86,1.5928) (0.87,1.5986) (0.88,1.60451) (0.89,1.61055) (0.9,1.61671) (0.91,1.623) (0.92,1.62944) (0.93,1.63602) (0.94,1.64276) (0.95,1.64966) (0.96,1.65675) (0.97,1.66404) (0.98,1.67155) (0.99,1.67931) (1.,1.68739)
    };
\end{axis}
\end{tikzpicture}}
~
\scalebox{0.55}{
\begin{tikzpicture}
\begin{axis}[
  xmin=0, xmax=1,
  ymin=1, ymax=1.7,
    ytick={1,1.1,...,1.8},
    ymajorgrids=true,
    xmajorgrids=true,
    grid style=dashed,
    xlabel = {$\lambda$},
    ylabel = {$\apx{20}{\lambda}$},
    thick
]
\addplot[
    color=red,
    ]
    coordinates {(0.001,1.0017) (0.01,1.01478) (0.02,1.0279) (0.03,1.04027) (0.04,1.05211) (0.05,1.06354) (0.06,1.07463) (0.07,1.08542) (0.08,1.09594) (0.09,1.10622) (0.1,1.11627) (0.11,1.12612) (0.12,1.13576) (0.13,1.14521) (0.14,1.15447) (0.15,1.16356) (0.16,1.17247) (0.17,1.18122) (0.18,1.18979) (0.19,1.1982) (0.2,1.20645) (0.21,1.21454) (0.22,1.22246) (0.23,1.23023) (0.24,1.23784) (0.25,1.24528) (0.26,1.25257) (0.27,1.25969) (0.28,1.26665) (0.29,1.27345) (0.3,1.28009) (0.31,1.28656) (0.32,1.29286) (0.33,1.299) (0.34,1.30497) (0.35,1.31076) (0.36,1.31638) (0.37,1.32183) (0.38,1.32709) (0.39,1.33218) (0.4,1.33708) (0.41,1.34179) (0.42,1.34631) (0.43,1.35064) (0.44,1.35478) (0.45,1.35872) (0.46,1.36245) (0.47,1.36598) (0.48,1.36929) (0.49,1.3724) (0.5,1.37528) (0.51,1.37799) (0.52,1.3807) (0.53,1.38344) (0.54,1.38622) (0.55,1.38902) (0.56,1.39186) (0.57,1.39473) (0.58,1.39763) (0.59,1.40056) (0.6,1.40353) (0.61,1.40653) (0.62,1.40957) (0.63,1.41264) (0.64,1.41575) (0.65,1.41889) (0.66,1.42207) (0.67,1.42529) (0.68,1.42855) (0.69,1.43185) (0.7,1.43519) (0.71,1.43857) (0.72,1.442) (0.73,1.44546) (0.74,1.44897) (0.75,1.45252) (0.76,1.45612) (0.77,1.45977) (0.78,1.46346) (0.79,1.4672) (0.8,1.47098) (0.81,1.47482) (0.82,1.47871) (0.83,1.48265) (0.84,1.48665) (0.85,1.4907) (0.86,1.4948) (0.87,1.49896) (0.88,1.50318) (0.89,1.50746) (0.9,1.5118) (0.91,1.5162) (0.92,1.52066) (0.93,1.52519) (0.94,1.52978) (0.95,1.53444) (0.96,1.53917) (0.97,1.54397) (0.98,1.54884) (0.99,1.55379) (1.,1.55881)
    };
\addplot[
    color=blue,
    ]
    coordinates {(0.001,1.2679) (0.01,1.26971) (0.02,1.27171) (0.03,1.2737) (0.04,1.27569) (0.05,1.27768) (0.06,1.27965) (0.07,1.28162) (0.08,1.28358) (0.09,1.28553) (0.1,1.28747) (0.11,1.28939) (0.12,1.2913) (0.13,1.2932) (0.14,1.29508) (0.15,1.29696) (0.16,1.29886) (0.17,1.30077) (0.18,1.30271) (0.19,1.30466) (0.2,1.30664) (0.21,1.30863) (0.22,1.31064) (0.23,1.31267) (0.24,1.31472) (0.25,1.31679) (0.26,1.31888) (0.27,1.321) (0.28,1.32313) (0.29,1.32529) (0.3,1.32747) (0.31,1.32967) (0.32,1.33189) (0.33,1.33414) (0.34,1.33641) (0.35,1.3387) (0.36,1.34102) (0.37,1.34336) (0.38,1.34573) (0.39,1.34813) (0.4,1.35055) (0.41,1.353) (0.42,1.35548) (0.43,1.35798) (0.44,1.36052) (0.45,1.36308) (0.46,1.36567) (0.47,1.3683) (0.48,1.37096) (0.49,1.37364) (0.5,1.37637) (0.51,1.37912) (0.52,1.38191) (0.53,1.38474) (0.54,1.3876) (0.55,1.3905) (0.56,1.39344) (0.57,1.39642) (0.58,1.39944) (0.59,1.40251) (0.6,1.40562) (0.61,1.40877) (0.62,1.41197) (0.63,1.41522) (0.64,1.41852) (0.65,1.42187) (0.66,1.42528) (0.67,1.42874) (0.68,1.43226) (0.69,1.43584) (0.7,1.43949) (0.71,1.4432) (0.72,1.44699) (0.73,1.45085) (0.74,1.45478) (0.75,1.4588) (0.76,1.4629) (0.77,1.46709) (0.78,1.47137) (0.79,1.47576) (0.8,1.48025) (0.81,1.48485) (0.82,1.48958) (0.83,1.49443) (0.84,1.49941) (0.85,1.50454) (0.86,1.50982) (0.87,1.51527) (0.88,1.52089) (0.89,1.5267) (0.9,1.53272) (0.91,1.53895) (0.92,1.54542) (0.93,1.55214) (0.94,1.55914) (0.95,1.56644) (0.96,1.57408) (0.97,1.58208) (0.98,1.59049) (0.99,1.59937) (1.,1.60881)
    };
\end{axis}
\end{tikzpicture}}
~
\scalebox{0.55}{
\begin{tikzpicture}
\begin{axis}[
  xmin=0, xmax=1,
  ymin=1, ymax=1.7,
    ytick={1,1.1,...,1.8},
    ymajorgrids=true,
    xmajorgrids=true,
    grid style=dashed,
    xlabel = {$\lambda$},
    ylabel = {$\apx{100}{\lambda}$},
    thick
]
\addplot[
    color=red,
    ]
    coordinates {(0.001,1.00186) (0.01,1.01588) (0.02,1.0298) (0.03,1.04286) (0.04,1.05533) (0.05,1.06733) (0.06,1.07894) (0.07,1.09022) (0.08,1.10121) (0.09,1.11192) (0.1,1.12238) (0.11,1.13261) (0.12,1.14261) (0.13,1.15241) (0.14,1.162) (0.15,1.1714) (0.16,1.1806) (0.17,1.18962) (0.18,1.19846) (0.19,1.20711) (0.2,1.21559) (0.21,1.2239) (0.22,1.23203) (0.23,1.23998) (0.24,1.24777) (0.25,1.25538) (0.26,1.26281) (0.27,1.27008) (0.28,1.27717) (0.29,1.28409) (0.3,1.29083) (0.31,1.2974) (0.32,1.30379) (0.33,1.31) (0.34,1.31603) (0.35,1.32188) (0.36,1.32754) (0.37,1.33302) (0.38,1.33831) (0.39,1.34341) (0.4,1.34831) (0.41,1.35302) (0.42,1.35753) (0.43,1.36183) (0.44,1.36594) (0.45,1.36983) (0.46,1.37351) (0.47,1.37698) (0.48,1.38023) (0.49,1.38326) (0.5,1.38608) (0.51,1.38888) (0.52,1.39172) (0.53,1.39458) (0.54,1.39748) (0.55,1.40041) (0.56,1.40337) (0.57,1.40637) (0.58,1.4094) (0.59,1.41246) (0.6,1.41556) (0.61,1.41869) (0.62,1.42186) (0.63,1.42507) (0.64,1.42831) (0.65,1.4316) (0.66,1.43492) (0.67,1.43828) (0.68,1.44168) (0.69,1.44512) (0.7,1.44861) (0.71,1.45213) (0.72,1.45571) (0.73,1.45932) (0.74,1.46298) (0.75,1.46668) (0.76,1.47044) (0.77,1.47424) (0.78,1.47808) (0.79,1.48198) (0.8,1.48593) (0.81,1.48993) (0.82,1.49398) (0.83,1.49809) (0.84,1.50225) (0.85,1.50647) (0.86,1.51075) (0.87,1.51508) (0.88,1.51947) (0.89,1.52393) (0.9,1.52844) (0.91,1.53302) (0.92,1.53767) (0.93,1.54238) (0.94,1.54716) (0.95,1.55201) (0.96,1.55693) (0.97,1.56193) (0.98,1.567) (0.99,1.57214) (1.,1.57737)
    };
\addplot[
    color=blue,
    ]
    coordinates {(0.001,1.2303) (0.01,1.23332) (0.02,1.2367) (0.03,1.24009) (0.04,1.2435) (0.05,1.24692) (0.06,1.25036) (0.07,1.25381) (0.08,1.25726) (0.09,1.26072) (0.1,1.26419) (0.11,1.26767) (0.12,1.27115) (0.13,1.27462) (0.14,1.2781) (0.15,1.28157) (0.16,1.28504) (0.17,1.2885) (0.18,1.29195) (0.19,1.29539) (0.2,1.29882) (0.21,1.30222) (0.22,1.30561) (0.23,1.30897) (0.24,1.31231) (0.25,1.31562) (0.26,1.3189) (0.27,1.32214) (0.28,1.32535) (0.29,1.32851) (0.3,1.33164) (0.31,1.33471) (0.32,1.33773) (0.33,1.3407) (0.34,1.3436) (0.35,1.34645) (0.36,1.34923) (0.37,1.35193) (0.38,1.35456) (0.39,1.35711) (0.4,1.35961) (0.41,1.36214) (0.42,1.36469) (0.43,1.36726) (0.44,1.36987) (0.45,1.3725) (0.46,1.37516) (0.47,1.37784) (0.48,1.38056) (0.49,1.3833) (0.5,1.38608) (0.51,1.38888) (0.52,1.39172) (0.53,1.39458) (0.54,1.39748) (0.55,1.40041) (0.56,1.40337) (0.57,1.40637) (0.58,1.4094) (0.59,1.41246) (0.6,1.41556) (0.61,1.41869) (0.62,1.42186) (0.63,1.42507) (0.64,1.42831) (0.65,1.4316) (0.66,1.43492) (0.67,1.43828) (0.68,1.44168) (0.69,1.44512) (0.7,1.44861) (0.71,1.45213) (0.72,1.45571) (0.73,1.45932) (0.74,1.46298) (0.75,1.46668) (0.76,1.47044) (0.77,1.47424) (0.78,1.47808) (0.79,1.48198) (0.8,1.48593) (0.81,1.48993) (0.82,1.49398) (0.83,1.49809) (0.84,1.50225) (0.85,1.50647) (0.86,1.51075) (0.87,1.51508) (0.88,1.51948) (0.89,1.52394) (0.9,1.52847) (0.91,1.53307) (0.92,1.53776) (0.93,1.54255) (0.94,1.54747) (0.95,1.55259) (0.96,1.55799) (0.97,1.56386) (0.98,1.57045) (0.99,1.57813) (1.,1.58737)
    };
\end{axis}
\end{tikzpicture}}
\Description{Upper and lower bounds on the approximation ratio of anonymous pricing for $n=5$ (left), $n=20$ (centre) and $n=100$ (right) bidders with i.i.d. $\lambda$-regular valuations, as a function of $\lambda$.}
\caption{Upper (\cref{thm:regularuppfix}) and lower (\cref{thm:regularlowfix}) bounds on the approximation ratio of anonymous pricing for $n=5$ (left), $n=20$ (centre) and $n=100$ (right) bidders with i.i.d.\ $\lambda$-regular valuations, as a function of $\lambda$.}
\label{fig:regupplowbounds}
\end{figure}
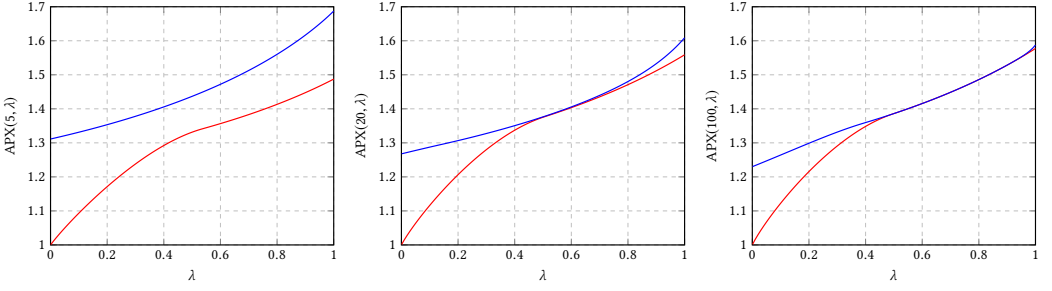

In this section we provide the generalized, $\lambda$-regular counterparts of our
main results for MHR distributions of the previous~\cref{sec:MHR}. In particular, we
provide upper (\cref{thm:regularuppfix}; cf.~\cref{th:main_upper}) and lower
(\cref{thm:regularlowfix}; cf.~\cref{th:lower_expo}) bounds on the approximation
ratio of anonymous pricing for $n$ bidders with i.i.d.\ $\lambda$-regular
valuations, and also derive an asymptotically tight expression for large $n$
(\cref{thm:regularinf}; cf.~\cref{th:main_upper}).

The following quantities will help us to state and prove our results in this section. 
For an integer $n\geq 2$ and
$\lambda\in(0,1]$, we define the quantities $\beta_{n,\lambda}, a_{n,\lambda} > 0$
and the function $g_{n,\lambda}:[0,\infty)\rightarrow[0,\infty)$ as follows:
\begin{equation}
\label{eq:helperreg_ab}
    \beta_{n,\lambda} \equiv \frac{n!\Gamma(2-\lambda)}{\Gamma(n+1-\lambda)}=\prod_{k=2}^n\left(1-\frac{\lambda}{k}\right)^{-1},
    \qquad
    a_{n,\lambda} \equiv
      \begin{cases}
      \frac{(1-\lambda)^{1/\lambda}\left[1-(1-\lambda)^{1/\lambda}\right]^{n-1}}{1-\left[1-(1-\lambda)^{1/\lambda}\right]^n},
      &\text{if}\;\; \lambda\in (0,1)\;,\\
      \frac{1}{n}, &\text{if}\;\; \lambda=1\;,
      \end{cases}
\end{equation}
\begin{equation}
\label{eq:helperreg_g}
    g_{n,\lambda}(c) \equiv c\left[1-\left[1-\left[1+c\left(\beta_{n,\lambda}-1\right)\right]^{-1/\lambda}\right]^n\right]\;.
\end{equation}
The function $g_n$ introduced in \cref{eq:gdef} can be recovered from $g_{n,\lambda}$ in the limit $\lambda\rightarrow0$. Note also that $\beta_{n,\lambda}$ corresponds to the fraction in
\cref{lemma:bound_tail_regular} with $k=n-1$. Using Stirling's approximation (see,
e.g.,~\citet[Ch. 8]{rudin:76}) we can derive the asymptotics
\begin{equation}\label{eq:betaapprox}\beta_{n,\lambda}=\frac{n!\Gamma(2-\lambda)}{\Gamma(n+1-\lambda)}\sim\Gamma(2-\lambda)n^\lambda\;,\end{equation}
valid for fixed $\lambda$, and for large $n$.

\begin{theorem}
\label{thm:regularuppfix}
Let $F$ be a $\lambda$-regular distribution with $\lambda\in(0,1]$, and $n\geq 2$.
Using the same take-it-or-leave-it price, to sell an item to $n$ buyers with i.i.d.\
valuations from $F$, achieves an approximation ratio of at most
\[\apx{F}{n}\leq \frac{1}{\sup_{c\in[0,1]}g_{n,\lambda}(c)}+a_{n,\lambda}\]
with respect to the optimal revenue,
where $a_{n,\lambda}$ and $g_{n,\lambda}$ are defined as in \eqref{eq:helperreg_ab} and \eqref{eq:helperreg_g}.
\end{theorem}

\begin{proof}
We follow the same steps as for \cref{th:main_upper}, but use
\cref{lemma:bound_tail_regular} and the quantile bound from
\cref{lemma:regularproperties}. Specifically, instead of \eqref{eq:bound_price_second_order_c} we bound the revenue from pricing at $c\cdot \nu$, where $c\in[0,1]$ and $\nu=\expect{\orderstat{X}{n-1}{n}}$, as
\begin{equation}
\label{eq:bound_price_second_order_c_lambda}
\pricep{F}{n}{c\nu}
=c\nu[1-F(c\nu)^n]
\geq c\nu\left[1-\left[1-\left[1+c\left(\beta_{n,\lambda}-1\right)\right]^{-1/\lambda}\right]^n\right]\;.
\end{equation}
the inequality holding due to \cref{lemma:bound_tail_regular} (for $k=n-1$). Optimizing with respect to $c$ we get that
$$\price{F}{n} \geq \nu \sup_{c\in[0,1]}g_{n,\lambda}(c)\;.
$$
We can now plug this bound at the derivation in \cref{th:main_upper} to get
\begin{align*}
\apx{F}{n}&= \frac{\myerson{F}{n}}{\price{F}{n}} \\
  &\leq \frac{\nu}{\nu \sup_{c\in[0,1]}g_{n,\lambda}(c)}+ \frac{R(q^*)(1-q^*)^{n-1}}{R(q^*)\frac{1-(1-q^*)^n}{q^*}} \\
  &= \frac{1}{\sup_{c\in[0,1]}g_{n,\lambda}(c)} + \frac{q^*(1-q^*)^{n-1}}{1-(1-q^*)^n} \\
  & \leq \frac{1}{\sup_{c\in[0,1]}g_{n,\lambda}(c)} + a_{n,\lambda}\;.
\end{align*}
For the last step, we use the bound $q^\ast\geq (1-\lambda)^{1/\lambda}$ from \cref{lemma:regularproperties}, together with the definition of $a_{n,\lambda}$ as in \eqref{eq:helperreg_ab}.
\end{proof}

For the lower bounds, we introduce the following family of \emph{rescaled Pareto} distributions $F_{\lambda,r}$, with support $[1,\infty)$, for $\lambda\in(0,1]$ and $r\in(0,1]$:

\begin{equation}
\label{eq:pareto_instances}
F_{\lambda,r}(x)=1-\frac{1}{\left[1+\frac{1}{r}(x-1)\right]^{1/\lambda}}\;.
\end{equation}
These can be seen as lying at the edge of
$\lambda$-regularity; it is not difficult to see that $F_{\lambda, r}$ is $\lambda$-regular, but \emph{not} $\lambda'$-regular for $\lambda'<\lambda$ (see~\cref{sec:lambdareg-model}).
They are inspired by the lower-bound construction of~
\citet[Appendix~C.3]{Duetting2016}, but we generalized them to be $\lambda$-regular and also removed the point mass at the right endpoint of the
support in order to guarantee continuity. 
On the other hand, for $r=1$, our distributions reduce to standard Pareto: $F_{\lambda,1}=1-\frac{1}{x^{1/\lambda}}$.  

We first collect below some basic properties of rescaled Pareto distributions, concerning
their expected second-highest order statistic and pricing revenue;  their proof can be found
in \cref{app:auxfunprops}.
\begin{lemma}\label{lem:basicpareto}
For any $\lambda\in(0,1]$ and $r\in(0,1]$,
\begin{enumerate}
    \item $\expect[X\sim F_{\lambda,r}]{X_{n-1:n}}=1+r\left(\beta_{n,\lambda}-1\right)$;
    \item $\price{F_{\lambda,r}}{n}=\sup_{q\in[0,1]}\left(1+r\left(\frac{1}{q^\lambda}-1\right)\right)\left(1-(1-q)^n\right)$.
\end{enumerate}
\end{lemma}

For the sake of exposition we introduce the function $H_{n,\lambda}(r,q)$, for
$n\geq2$, $\lambda\in(0,1]$, $r\in(0,1]$, and $q=[0,1]$,
\begin{equation}\label{eq:hnlambda}
    H_{n,\lambda}(r,q)=\left(1+r\left(\frac{1}{q^\lambda}-1\right)\right)\left(1-(1-q)^n\right)\;.
\end{equation}
Note that $H_{n,\lambda}(r,q)$ is continuously defined at $q=0$ as the singularity
is removable.

\begin{theorem}\label{thm:regularlowfix} For $\lambda\in(0,1]$, $r\in(0,1]$, and
$n\geq 2$ bidders with i.i.d.\ valuations from the rescaled Pareto distribution
$F_{\lambda,r}$ (see~\eqref{eq:pareto_instances}), the approximation ratio of anonymous pricing is at least

\begin{equation}\label{eq:regularlowfix_1}\apx{F_{\lambda,r}}{n}\geq
\frac{1+r\left(\beta_{n,\lambda}-1\right)}{\sup_{q\in[0,1]}H_{n,\lambda}(r,q)}\;.\end{equation}
This implies the following lower bound on the approximation ratio of $\lambda$-regular
distributions:
\begin{equation}\label{eq:regularlowfix}\apx{n}{\lambda}\geq\adjustlimits\sup_{r\in[0,1]}\inf_{q\in[0,1]}\frac{1+r\left(\beta_{n,\lambda}-1\right)}{H_{n,\lambda}(r,q)}\;.\end{equation}
\end{theorem}

\begin{proof}
Using \cref{lem:basicpareto} (1), and the fact that the optimal revenue can be lower bounded by the second-highest order statistic, gives
\[\myerson{F_{\lambda,r}}{n}\geq\expect[X\sim F_{\lambda,r}]{X_{n-1:n}}=1+r\left(\beta_{n,\lambda}-1\right)\;.\]
Using \cref{lem:basicpareto} (2), and the definition of $H_{n,\lambda}$ from \eqref{eq:hnlambda}, gives
\[\price{F_{\lambda,r}}{n}=\sup_{q\in[0,1]}H_{n,\lambda}(r,q)\;.\]
Taking the ratio between these two quantities proves \eqref{eq:regularlowfix_1}. If we optimize this ratio for $r\in(0,1]$, we obtain
\begin{align*}
\apx{n}{\lambda}&\geq\sup_{r\in(0,1]}\apx{F_{\lambda,r}}{n}\\
&\geq\sup_{r\in(0,1]}\frac{1+r\left(\beta_{n,\lambda}-1\right)}{\sup_{q\in[0,1]}H_{n,\lambda}(r,q)}\\
&=\adjustlimits\sup_{r\in[0,1]}\inf_{q\in[0,1]}\frac{1+r\left(\beta_{n,\lambda}-1\right)}{H_{n,\lambda}(r,q)}\;;
\end{align*}
note that in the last step, we can allow $r=0$ without loss since the right-hand side is well-defined and gives a trivial lower bound of 1.
\end{proof}

A plot of the bounds given in
\cref{thm:regularuppfix,thm:regularlowfix} can be seen in
\cref{fig:regupplowbounds}, for different values of $n$ and $\lambda$.

\subsection{Asymptotic Analysis}
\label{sec:lambda_reg_asy}
\begin{figure}
\centering
\begin{tikzpicture}
\begin{axis}[
  xmin=0, xmax=1,
  ymin=1, ymax=1.6,
    ytick={1,1.1,...,1.6},
    ymajorgrids=true,
    xmajorgrids=true,
    extra y ticks={1.58198},
    grid style=dashed,
    xlabel = {$\lambda$},
    ylabel = {$\apx{\infty}{\lambda}$},
    thick
]
\addplot[
    color=blue,
    ]
    coordinates {(0,1) (0.01,1.01616) (0.02,1.03028) (0.03,1.04351) (0.04,1.05614) (0.05,1.06828) (0.06,1.08003) (0.07,1.09143) (0.08,1.10253) (0.09,1.11334) (0.1,1.12391) (0.11,1.13423) (0.12,1.14432) (0.13,1.1542) (0.14,1.16388) (0.15,1.17335) (0.16,1.18262) (0.17,1.19171) (0.18,1.20061) (0.19,1.20933) (0.2,1.21786) (0.21,1.22622) (0.22,1.2344) (0.23,1.2424) (0.24,1.25023) (0.25,1.25788) (0.26,1.26535) (0.27,1.27265) (0.28,1.27977) (0.29,1.28672) (0.3,1.29349) (0.31,1.30008) (0.32,1.30649) (0.33,1.31272) (0.34,1.31876) (0.35,1.32462) (0.36,1.33029) (0.37,1.33578) (0.38,1.34107) (0.39,1.34617) (0.4,1.35107) (0.41,1.35578) (0.42,1.36028) (0.43,1.36458) (0.44,1.36867) (0.45,1.37256) (0.46,1.37623) (0.47,1.37968) (0.48,1.38291) (0.49,1.38592) (0.5,1.38874) (0.51,1.39158) (0.52,1.39444) (0.53,1.39734) (0.54,1.40027) (0.55,1.40323) (0.56,1.40622) (0.57,1.40925) (0.58,1.41231) (0.59,1.41541) (0.6,1.41854) (0.61,1.4217) (0.62,1.42491) (0.63,1.42815) (0.64,1.43143) (0.65,1.43475) (0.66,1.4381) (0.67,1.4415) (0.68,1.44494) (0.69,1.44841) (0.7,1.45193) (0.71,1.4555) (0.72,1.4591) (0.73,1.46276) (0.74,1.46645) (0.75,1.4702) (0.76,1.47399) (0.77,1.47783) (0.78,1.48171) (0.79,1.48565) (0.8,1.48964) (0.81,1.49368) (0.82,1.49777) (0.83,1.50192) (0.84,1.50613) (0.85,1.51039) (0.86,1.5147) (0.87,1.51908) (0.88,1.52352) (0.89,1.52802) (0.9,1.53258) (0.91,1.5372) (0.92,1.54189) (0.93,1.54665) (0.94,1.55148) (0.95,1.55638) (0.96,1.56135) (0.97,1.56639) (0.98,1.57151) (0.99,1.5767) (1.,1.58198)
    };
\end{axis}
\end{tikzpicture}
\Description{The asymptotically tight value of the approximation ratio of anonymous
pricing, given in \cref{thm:regularinf}, as a function of the regularity parameter
$\lambda$ (for $n\to\infty$ i.i.d. bidders).}
\caption{ The asymptotically tight value of the approximation ratio of anonymous
pricing, given in \cref{thm:regularinf}, as a function of the regularity parameter
$\lambda$ (for $n\to\infty$ i.i.d.\ bidders). }
\label{fig:regular_asymptotic}
\end{figure}
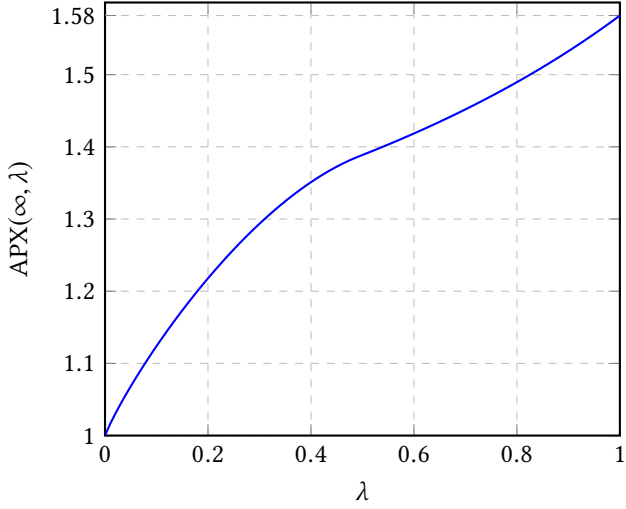

By observing the plots in \cref{fig:regupplowbounds}, it seems that the upper and
lower bounds are approaching some smooth function of $\lambda$ as $n$ grows large;
moreover, this function appears to increase from $1$ at $\lambda=0$ to
$\frac{e}{e-1}\approx 1.58$ at $\lambda=1$, which would recover the known bounds for
MHR (this paper, \cref{sec:MHR}) and regular (see,
e.g.,~\citep{Chawla2010a,Duetting2016}) distributions. In the remainder of this
paper we prove this is indeed the case, and characterize the limiting function. In
other words, we are interesting in taking the limit (as $n\rightarrow\infty$) of the
quantities defined in \cref{thm:regularuppfix,thm:regularlowfix}.

For each $\lambda\in(0,1]$, let us define the function $g_\lambda:[0,\infty)\rightarrow[0,\infty)$ as follows.

\begin{equation}\label{eq:glambda}
    g_{\lambda}(c)=c\left(1-e^{-(c\Gamma(2-\lambda))^{-1/\lambda}}\right)\;.
\end{equation}
The function $g_\lambda$ can be obtained from the function $g_{n,\lambda}$ introduced in \cref{eq:helperreg_g} as $n\rightarrow\infty$. We also define $\lambda^\ast\approx 0.4940$ as the unique positive root over $(0,1]$
of the equation (see~\cref{lem:lambdastar})
\begin{equation}\label{eq:lambdastar}1-\left(1+\frac{\Gamma(2-\lambda)^{-1/\lambda}}{\lambda}\right)e^{-\Gamma(2-\lambda)^{-1/\lambda}}=0\;.\end{equation}

For each $\lambda\in(0,1)$, let also $\eta(\lambda)$ denote the unique positive
solution of the equation $e^x=1+\frac{x}{\lambda}$, which is related to the function
$\eta(k)$ appearing in~\citet[Section 3.1.2]{Blumrosen2008a}. In fact, we must
mention here that one can derive a lower bound on the asymptotic approximation ratio
from~\citep{Blumrosen2008a}, which coincides with our \cref{thm:regularinf} in the
branch $0<\lambda\leq\lambda^\ast$. Although Blumrosen and Holenstein, in general,
study distributions satisfying the von Mises conditions, for their lower bounds they
make use of power law distributions of the form $F(x)=1-1/x^k$, for $k>1$. These
correspond to our rescaled Pareto distribution $F_{r,\lambda}$ from
\eqref{eq:pareto_instances} with $r=1$ and $\lambda=1/k$.

For the main result of this section we shall make use of the following technical lemma whose proof can be found in \cref{app:auxfunprops}.

\begin{lemma}\label{lem:supglam} For each $\lambda\in(0,1]$, let $g_\lambda$ be defined as in \eqref{eq:glambda} and $\eta(\lambda)$ be the unique positive solution of the equation $e^x=1+\frac{x}{\lambda}$. Let also $\lambda^\ast$ be the unique root of \eqref{eq:lambdastar}. The supremum of the function $g_\lambda$ over $[0,1]$ is as follows.
\begin{align*}
\text{If }0<\lambda\leq\lambda^\ast\text{, then}\quad&\sup_{c\in[0,1]}g_\lambda(c)=\sup_{c\geq 0}g_\lambda(c)=\frac{\eta(\lambda)^{1-\lambda}}{\Gamma(2-\lambda)(\lambda+\eta(\lambda))}\;.\\
\text{If }\lambda^\ast\leq\lambda\leq 1\text{, then}\quad&\sup_{c\in[0,1]}g_\lambda(c)=\phantom{\sup_{c\geq0}}g_\lambda(1)=1-e^{-\Gamma(2-\lambda)^{-1/\lambda}}\;.
\end{align*}
\end{lemma}

\begin{theorem}\label{thm:regularinf}
For $\lambda\in(0,1]$, the approximation ratio of anonymous pricing with arbitrarily
many bidders having i.i.d.\ valuations from a $\lambda$-regular distribution, is
asymptotically
\[\lim_{n\rightarrow\infty}\apx{n}{\lambda}=
\frac{1}{\sup_{0\leq c\leq1}g_\lambda(c)}=
\left\{\begin{array}{cc}\frac{\Gamma(2-\lambda)(\lambda+\eta(\lambda))}{\eta(\lambda)^{1-\lambda}},&0<\lambda\leq\lambda^\ast\;,\\
\\\frac{1}{1-e^{-\Gamma(2-\lambda)^{-1/\lambda}}},&\lambda^\ast\leq\lambda\leq
1\;,\end{array}\right.\] where $g_\lambda$ is defined in \eqref{eq:glambda},
$\lambda^\ast$ is the unique root of \eqref{eq:lambdastar}, $\Gamma$ is the gamma
function, and $\eta(\lambda)$ is the unique positive solution of the equation
$e^x=1+\frac{x}{\lambda}$.
\end{theorem}
A plot of this asymptotic approximation ratio can be seen
in~\cref{fig:regular_asymptotic}.
\begin{proof} The second equality, i.e.\ the characterization of $\sup_{c\in[0,1]}g_\lambda(c)$, comes from \cref{lem:supglam}. 
The proof of the theorem can be split into three parts; we
start by providing an upper bound on the asymptotics; as for the lower bounds, our
proof considers the cases  $0<\lambda\leq\lambda^\ast$ and
$\lambda^\ast<\lambda\leq1$  separately.

To prove an upper bound, we start from the result obtained in
\cref{thm:regularuppfix},

\[\apx{n}{\lambda}\leq
\frac{1}{\sup_{c\in[0,1]}g_{n,\lambda}(c)}+a_{n,\lambda}\;.\]

For each $\lambda\in(0,1]$ and each $c\in[0,1]$, the Stirling approximation in
\eqref{eq:betaapprox} yields the pointwise convergence $g_{n,\lambda}(c)\rightarrow
g_\lambda(c)$. Thus, by elementary analysis\footnote{It is worth
mentioning that one could prove, with some technical effort, that the convergence
$g_{n,\lambda}(c)\rightarrow g_\lambda(c)$ is actually uniform on $c$, which would
allow us to interchange the limit with the supremum and write
$\lim_{n\to\infty}\sup_{c\in[0,1]}g_{n,\lambda}(c)=\sup_{c\in[0,1]}g_{\lambda}(c)$;
however, we will not need this result as the lower bound will be matching, so we
omit its proof.}, we have
$$\liminf_{n\rightarrow\infty}\sup_{0\leq c\leq
1}g_{n,\lambda}(c)\geq\sup_{0\leq c\leq 1}g_{\lambda}(c)\;.$$  
As the additive term $a_{n,\lambda}$ vanishes as
$n\to\infty$, we get the desired upper bound,
\begin{equation}\label{eq:reguppasym}
\limsup_{n\rightarrow\infty}\apx{n}{\lambda}\leq\frac{1}{\sup_{c\in[0,1]}g_\lambda(c)}\;.
\end{equation}

To prove lower bounds, we study the asymptotic behaviour of the result obtained in
\cref{thm:regularlowfix},
\[\apx{n}{\lambda}\geq\adjustlimits\sup_{r\in[0,1]}\inf_{q\in[0,1]}\frac{1+r\left(\beta_{n,\lambda}-1\right)}{H_{n,\lambda}(r,q)}\;.\]

We begin by setting $r=1$, and applying the change of variables
$c=1/(\Gamma(2-\lambda)n^\lambda q^\lambda)$, yielding
\begin{align*}
\apx{n}{\lambda}&\geq\frac{\beta_{n,\lambda}}{\sup_{0\leq q\leq
1}\frac{1-(1-q)^n}{q^\lambda}}\\
&=\frac{\frac{\beta_{n,\lambda}}{\Gamma(2-\lambda)n^\lambda}}{\sup_{c\geq1/(\Gamma(2-\lambda)n^\lambda)}c\left[1-\left(1-\frac{(c\Gamma(2-\lambda))^{-1/\lambda}}{n}\right)^n\right]}\;.
\end{align*}
Note that the numerator of the above fraction goes to 1 as $n\rightarrow\infty$ due
to Stirling's approximation. As for the denominator, we need some technical work to
ensure that the limit and the supremum operators commute, and to show that it
converges to $\sup_{c\geq 0}g_\lambda(c)$. This is done in
\cref{lem:lowboundtechnical1} in \cref{app:auxfunprops}. Thus, we get
\begin{equation}
\label{eq:lower_asy}
\liminf_{n\rightarrow\infty}\apx{n}{\lambda}\geq\frac{1}{\sup_{c\geq
0}g_\lambda(c)}\;,
\end{equation}
which matches the desired quantity as long as $0<\lambda\leq\lambda^\ast$, due to \cref{lem:supglam}.

To get tight lower bounds for the case $\lambda^\ast<\lambda\leq 1$, we take an
approach very much inspired by that in \cite{Duetting2016}; our starting point is
again \cref{thm:regularlowfix}, but now the idea is to carefully choose a value of
$r$ which depends on $n$, for which a specific choice of optimal quantile $q^\ast$
gives the desired bound.

For each $n\geq2$ and $\lambda\in(0,1]$, recall the function
$H_{n,\lambda}:[0,1]\times[0,1]\rightarrow\mathbb{R}$ from
\eqref{eq:hnlambda}. As this function is continuous on a compact set, we can apply
Berge's Maximum Theorem (see, e.g,~\cite[Theorem~17.31]{Aliprantis2006a}) to
conclude that the correspondence $r\mapsto\argmax_{q\in[0,1]} H_{n,\lambda}(r,q)$ is
non-empty, compact valued and upper hemicontinuous.\footnote{One could, with
additional technical work, also prove that $H_{n,\lambda}(r,q)$ has a single peak
(in $q$) for each $r$, making this correspondence a continuous function; the rest of
the argument would be essentially an application of the intermediate value theorem.}
When $r=0$, we have $H_{n,\lambda}(0,q)=1-(1-q)^n$, which has a single peak at
$q=1$. Next we consider the case $r=1$, that is,
$H_{n,\lambda}(1,q)=\frac{1}{q^\lambda}(1-(1-q)^n)$. In
\cref{lem:lowboundtechnical2}, \cref{app:auxfunprops}, we prove that this function
has a single peak, and that for $\lambda^\ast<\lambda\leq 1$ and large enough $n$,
its peak lies to the left of $\beta_{n,\lambda}^{-1/\lambda}$. Thus, by
hemicontinuity, we can conclude that, for large enough $n$, there must exist some
value of $r\in[0,1]$, say $r_n$, such that $H_{n,\lambda}(r_n,q)$ is maximized at
precisely the quantile $q^\ast_n=\beta_{n,\lambda}^{-1/\lambda}$.\footnote{Note that
this statement is not merely existential: we can compute $r_n$ by solving the
equation $\frac{\partial}{\partial\, q} H_{n,\lambda}(r_n,q^\ast_n)=0$.} We can now
plug this into \cref{thm:regularlowfix}, obtaining
\begin{align*}
    \apx{n}{\lambda}&\geq\inf_{0\leq q\leq
    1}\frac{1+r_n\left(\beta_{n,\lambda}-1\right)}{H_{n,\lambda}(r_n,q)}\\
    &=\frac{1+r_n\left(\beta_{n,\lambda}-1\right)}{\left(1+r_n\left(\beta_{n,\lambda}-1\right)\right)\left(1-\left(1-\beta_{n,\lambda}^{-1/\lambda}\right)^n\right)}\\
    &=\frac{1}{1-\left(1-\beta_{n,\lambda}^{-1/\lambda}\right)^n}\;;
\end{align*} 
by taking limits (and with the help of Stirling's approximation), we get
\[\liminf_{n\rightarrow\infty}\apx{n}{\lambda}\geq\frac{1}{1-e^{-\Gamma(2-\lambda)^{-1/\lambda}}}\;,\]
which matches the desired quantity as long as $\lambda^\ast<\lambda\leq 1$, due to \cref{lem:supglam}. This concludes the proof.
\end{proof}

\section{Conclusion and Future Directions}
\label{sec:conclusion}

In this paper, we studied the performance of anonymous selling mechanisms for a
Bayesian setting with a single item and $n$ bidders with i.i.d.~ valuations. We
completely characterized the asymptotic approximation ratio (with respect to the
optimal, Myersonian revenue) for $\lambda$-regular distributions with
$\lambda\in[0,1]$,
\[\adjustlimits\lim_{n\rightarrow\infty}\sup_{F\in\mathcal{D}_\lambda}\apx{F}{n}=\frac{1}{\sup_{0\leq
c\leq 1}g_\lambda(c)}\;.\]
This quantity increases smoothly (see~\cref{fig:regular_asymptotic}) from $1$ for
MHR distributions to $\frac{e}{e-1}\approx 1.58$ for regular distributions. For the
special case of MHR distributions, we also characterized the rate of convergence and
provided explicit, ``good'' pricing schemes that depend only on the knowledge of the
expected second-highest order statistic of the distribution.

A few interesting questions remain. Two quantities of independent interest are
\[\sup_{n\geq
2}\sup_{F\in\mathcal{D}_\lambda}\apx{F}{n}\qquad\text{and}\qquad\adjustlimits\sup_{F\in\mathcal{D}_\lambda}\limsup_{n\rightarrow\infty}\apx{F}{n}\;.\]

The first one corresponds to \emph{global} bounds on the approximation ratio, which
require a finer analysis, especially in the ``small'' $n$ range. Even for MHR
distributions, there is still a gap between the global upper bound of $1.35$ (numerically
achieved at $n=3$ in \cref{th:main_upper_small}) and the global lower
bound of $1.27$ (numerically achieved at $n=17$ in \cref{th:lower_expo}).

The second one intuitively captures settings where the seller's prior knowledge of a
large market is assumed to be size-independent. Studying this quantity amounts
to asking: do the lower bounds of \cref{thm:regularinf} \emph{require} constructing
a ``bad'' distribution $F_n$ for \emph{each} $n$ separately? In the range closer to
MHR, i.e. $0<\lambda\leq\lambda^\ast$, the answer is ``no'': our lower bounds were
obtained by setting $r=1$ globally, independent of $n$. But even for regular
distributions, there is still a gap between the size-independent lower bound of
$1.40$  (achieved by numerically maximizing \eqref{eq:lower_asy}) and the upper
bound of $1.58$ (see~\cref{fig:regular_asymptotic} at $\lambda=1$).

Another interesting direction would be to characterize analytically the rate of
convergence (with respect to the number of bidders $n$) of the approximation ratio
given by~\cref{thm:regularuppfix}. In other words, generalize \cref{th:main_upper}
for $\lambda$-regular distributions with $\lambda>0$.

\bibliographystyle{ACM-Reference-Format}
\bibliography{pricing}


\begin{thebibliography}{49}


\ifx \showCODEN    \undefined \def \showCODEN     #1{\unskip}     \fi
\ifx \showDOI      \undefined \def \showDOI       #1{#1}\fi
\ifx \showISBNx    \undefined \def \showISBNx     #1{\unskip}     \fi
\ifx \showISBNxiii \undefined \def \showISBNxiii  #1{\unskip}     \fi
\ifx \showISSN     \undefined \def \showISSN      #1{\unskip}     \fi
\ifx \showLCCN     \undefined \def \showLCCN      #1{\unskip}     \fi
\ifx \shownote     \undefined \def \shownote      #1{#1}          \fi
\ifx \showarticletitle \undefined \def \showarticletitle #1{#1}   \fi
\ifx \showURL      \undefined \def \showURL       {\relax}        \fi
\providecommand\bibfield[2]{#2}
\providecommand\bibinfo[2]{#2}
\providecommand\natexlab[1]{#1}
\providecommand\showeprint[2][]{arXiv:#2}

\bibitem[\protect\citeauthoryear{Aggarwal, Goel, and Mehta}{Aggarwal
  et~al\mbox{.}}{2009}]%
        {aggarwal2009efficiency}
\bibfield{author}{\bibinfo{person}{Gagan Aggarwal}, \bibinfo{person}{Gagan
  Goel}, {and} \bibinfo{person}{Aranyak Mehta}.}
  \bibinfo{year}{2009}\natexlab{}.
\newblock \showarticletitle{Efficiency of (Revenue-) Optimal Mechanisms}. In
  \bibinfo{booktitle}{\emph{Proceedings of the 10th ACM conference on
  Electronic Commerce (EC)}}. \bibinfo{pages}{235--242}.
\newblock
\urldef\tempurl%
\url{https://doi.org/10.1145/1566374.1566408}
\showDOI{\tempurl}


\bibitem[\protect\citeauthoryear{Alaei}{Alaei}{2014}]%
        {Alaei2014}
\bibfield{author}{\bibinfo{person}{Saeed Alaei}.}
  \bibinfo{year}{2014}\natexlab{}.
\newblock \showarticletitle{Bayesian Combinatorial Auctions: Expanding Single
  Buyer Mechanisms to Many Buyers}.
\newblock \bibinfo{journal}{\emph{SIAM J. Comput.}} \bibinfo{volume}{43},
  \bibinfo{number}{2} (\bibinfo{year}{2014}), \bibinfo{pages}{930--972}.
\newblock
\urldef\tempurl%
\url{https://doi.org/10.1137/120878422}
\showDOI{\tempurl}


\bibitem[\protect\citeauthoryear{Alaei, Hartline, Niazadeh, Pountourakis, and
  Yuan}{Alaei et~al\mbox{.}}{2015}]%
        {Alaei2015}
\bibfield{author}{\bibinfo{person}{Saeed Alaei}, \bibinfo{person}{Jason
  Hartline}, \bibinfo{person}{Rad Niazadeh}, \bibinfo{person}{Emmanouil
  Pountourakis}, {and} \bibinfo{person}{Yang Yuan}.}
  \bibinfo{year}{2015}\natexlab{}.
\newblock \showarticletitle{Optimal Auctions vs. Anonymous Pricing}. In
  \bibinfo{booktitle}{\emph{Proceedings of the 56th Annual Symposium on
  Foundations of Computer Science (FOCS)}}. \bibinfo{pages}{1446--1463}.
\newblock
\urldef\tempurl%
\url{https://doi.org/10.1109/FOCS.2015.92}
\showDOI{\tempurl}


\bibitem[\protect\citeauthoryear{Aliprantis and Border}{Aliprantis and
  Border}{2006}]%
        {Aliprantis2006a}
\bibfield{author}{\bibinfo{person}{Charalambos~D. Aliprantis} {and}
  \bibinfo{person}{Kim~C. Border}.} \bibinfo{year}{2006}\natexlab{}.
\newblock \bibinfo{booktitle}{\emph{Infinite Dimensional Analysis: A
  Hitchhiker's Guide} (\bibinfo{edition}{3rd} ed.)}.
\newblock \bibinfo{publisher}{Springer-Verlag}.
\newblock
\urldef\tempurl%
\url{https://doi.org/10.1007/3-540-29587-9}
\showDOI{\tempurl}


\bibitem[\protect\citeauthoryear{Arnold, Balakrishnan, and Nagaraja}{Arnold
  et~al\mbox{.}}{2008}]%
        {Arnold_2008}
\bibfield{author}{\bibinfo{person}{Barry~C. Arnold}, \bibinfo{person}{N.
  Balakrishnan}, {and} \bibinfo{person}{H.~N. Nagaraja}.}
  \bibinfo{year}{2008}\natexlab{}.
\newblock \bibinfo{booktitle}{\emph{A First Course in Order Statistics}}.
\newblock \bibinfo{publisher}{SIAM}.
\newblock
\urldef\tempurl%
\url{https://doi.org/10.1137/1.9780898719062}
\showDOI{\tempurl}


\bibitem[\protect\citeauthoryear{Babaioff, Blumrosen, Dughmi, and
  Singer}{Babaioff et~al\mbox{.}}{2017}]%
        {babaioff2017posting}
\bibfield{author}{\bibinfo{person}{Moshe Babaioff}, \bibinfo{person}{Liad
  Blumrosen}, \bibinfo{person}{Shaddin Dughmi}, {and} \bibinfo{person}{Yaron
  Singer}.} \bibinfo{year}{2017}\natexlab{}.
\newblock \showarticletitle{Posting Prices with Unknown Distributions}.
\newblock \bibinfo{journal}{\emph{ACM Trans. Econ. Comput.}}
  \bibinfo{volume}{5}, \bibinfo{number}{2} (\bibinfo{year}{2017}),
  \bibinfo{pages}{13:1--13:20}.
\newblock
\urldef\tempurl%
\url{https://doi.org/10.1145/3037382}
\showDOI{\tempurl}


\bibitem[\protect\citeauthoryear{Babaioff, Dughmi, Kleinberg, and
  Slivkins}{Babaioff et~al\mbox{.}}{2015}]%
        {Babaioff2015}
\bibfield{author}{\bibinfo{person}{Moshe Babaioff}, \bibinfo{person}{Shaddin
  Dughmi}, \bibinfo{person}{Robert Kleinberg}, {and}
  \bibinfo{person}{Aleksandrs Slivkins}.} \bibinfo{year}{2015}\natexlab{}.
\newblock \showarticletitle{Dynamic Pricing with Limited Supply}.
\newblock \bibinfo{journal}{\emph{ACM Trans. Econ. Comput.}}
  \bibinfo{volume}{3}, \bibinfo{number}{1} (\bibinfo{year}{2015}),
  \bibinfo{pages}{1--26}.
\newblock
\urldef\tempurl%
\url{https://doi.org/10.1145/2559152}
\showDOI{\tempurl}


\bibitem[\protect\citeauthoryear{Barlow and Marshall}{Barlow and
  Marshall}{1964}]%
        {Barlow1964}
\bibfield{author}{\bibinfo{person}{Richard~E. Barlow} {and}
  \bibinfo{person}{Albert~W. Marshall}.} \bibinfo{year}{1964}\natexlab{}.
\newblock \showarticletitle{Bounds for Distributions with Monotone Hazard Rate,
  {I}}.
\newblock \bibinfo{journal}{\emph{The Annals of Mathematical Statistics}}
  \bibinfo{volume}{35}, \bibinfo{number}{3} (\bibinfo{year}{1964}),
  \bibinfo{pages}{1234--1257}.
\newblock
\urldef\tempurl%
\url{https://doi.org/10.1214/aoms/1177703281}
\showDOI{\tempurl}


\bibitem[\protect\citeauthoryear{Barlow and Proschan}{Barlow and
  Proschan}{1996}]%
        {Barlow1996}
\bibfield{author}{\bibinfo{person}{Richard~E. Barlow} {and}
  \bibinfo{person}{Frank Proschan}.} \bibinfo{year}{1996}\natexlab{}.
\newblock \bibinfo{booktitle}{\emph{Mathematical Theory of Reliability}}.
\newblock \bibinfo{publisher}{SIAM}.
\newblock
\urldef\tempurl%
\url{https://doi.org/10.1137/1.9781611971194}
\showDOI{\tempurl}


\bibitem[\protect\citeauthoryear{Bhattacharya, Goel, Gollapudi, and
  Munagala}{Bhattacharya et~al\mbox{.}}{2010}]%
        {Bhattacharya2010a}
\bibfield{author}{\bibinfo{person}{Sayan Bhattacharya}, \bibinfo{person}{Gagan
  Goel}, \bibinfo{person}{Sreenivas Gollapudi}, {and} \bibinfo{person}{Kamesh
  Munagala}.} \bibinfo{year}{2010}\natexlab{}.
\newblock \showarticletitle{Budget Constrained Auctions with Heterogeneous
  Items}. In \bibinfo{booktitle}{\emph{Proceedings of the 42nd ACM symposium on
  Theory of Computing}}. ACM, \bibinfo{pages}{379--388}.
\newblock
\urldef\tempurl%
\url{https://doi.org/10.1145/1806689.1806743}
\showDOI{\tempurl}


\bibitem[\protect\citeauthoryear{Blumrosen and Holenstein}{Blumrosen and
  Holenstein}{2008}]%
        {Blumrosen2008a}
\bibfield{author}{\bibinfo{person}{Liad Blumrosen} {and}
  \bibinfo{person}{Thomas Holenstein}.} \bibinfo{year}{2008}\natexlab{}.
\newblock \showarticletitle{Posted Prices vs. Negotiations: An Asymptotic
  Analysis}. In \bibinfo{booktitle}{\emph{Proceedings of the 9th ACM Conference
  on Electronic Commerce (EC)}}. \bibinfo{pages}{49}.
\newblock
\urldef\tempurl%
\url{https://doi.org/10.1145/1386790.1386801}
\showDOI{\tempurl}


\bibitem[\protect\citeauthoryear{Bulow and Klemperer}{Bulow and
  Klemperer}{1996}]%
        {Bulow1996}
\bibfield{author}{\bibinfo{person}{Jeremy Bulow} {and} \bibinfo{person}{Paul
  Klemperer}.} \bibinfo{year}{1996}\natexlab{}.
\newblock \showarticletitle{Auctions Versus Negotiations}.
\newblock \bibinfo{journal}{\emph{The American Economic Review}}
  \bibinfo{volume}{86}, \bibinfo{number}{1} (\bibinfo{year}{1996}),
  \bibinfo{pages}{180--194}.
\newblock
\urldef\tempurl%
\url{http://www.jstor.org/stable/2118262}
\showURL{%
\tempurl}


\bibitem[\protect\citeauthoryear{Cai and Daskalakis}{Cai and
  Daskalakis}{2011}]%
        {Cai2011b}
\bibfield{author}{\bibinfo{person}{Yang Cai} {and}
  \bibinfo{person}{Constantinos Daskalakis}.} \bibinfo{year}{2011}\natexlab{}.
\newblock \showarticletitle{Extreme-Value Theorems for Optimal Multidimensional
  Pricing}.
\newblock \bibinfo{journal}{\emph{Proceedings of the 52nd Annual Symposium on
  Foundations of Computer Science (FOCS)}} (\bibinfo{year}{2011}),
  \bibinfo{pages}{522--531}.
\newblock
\urldef\tempurl%
\url{https://doi.org/10.1109/FOCS.2011.76}
\showDOI{\tempurl}


\bibitem[\protect\citeauthoryear{Caplin and Nalebuff}{Caplin and
  Nalebuff}{1991}]%
        {Caplin:1991aa}
\bibfield{author}{\bibinfo{person}{Andrew Caplin} {and} \bibinfo{person}{Barry
  Nalebuff}.} \bibinfo{year}{1991}\natexlab{}.
\newblock \showarticletitle{Aggregation and Social Choice: A Mean Voter
  Theorem}.
\newblock \bibinfo{journal}{\emph{Econometrica}} \bibinfo{volume}{59},
  \bibinfo{number}{1} (\bibinfo{year}{1991}), \bibinfo{pages}{1--24}.
\newblock
\urldef\tempurl%
\url{https://doi.org/10.2307/2938238}
\showDOI{\tempurl}


\bibitem[\protect\citeauthoryear{Cesa-Bianchi, Gentile, and
  Mansour}{Cesa-Bianchi et~al\mbox{.}}{2015}]%
        {Cesa-Bianchi2015}
\bibfield{author}{\bibinfo{person}{Nicolo Cesa-Bianchi},
  \bibinfo{person}{Claudio Gentile}, {and} \bibinfo{person}{Yishay Mansour}.}
  \bibinfo{year}{2015}\natexlab{}.
\newblock \showarticletitle{Regret Minimization for Reserve Prices in
  Second-Price Auctions}.
\newblock \bibinfo{journal}{\emph{{IEEE} Transactions on Information Theory}}
  \bibinfo{volume}{61}, \bibinfo{number}{1} (\bibinfo{year}{2015}),
  \bibinfo{pages}{549--564}.
\newblock
\urldef\tempurl%
\url{https://doi.org/10.1109/tit.2014.2365772}
\showDOI{\tempurl}


\bibitem[\protect\citeauthoryear{Chawla, Hartline, Malec, and Sivan}{Chawla
  et~al\mbox{.}}{2010}]%
        {Chawla2010a}
\bibfield{author}{\bibinfo{person}{Shuchi Chawla}, \bibinfo{person}{Jason~D
  Hartline}, \bibinfo{person}{David~L. Malec}, {and}
  \bibinfo{person}{Balasubramanian Sivan}.} \bibinfo{year}{2010}\natexlab{}.
\newblock \showarticletitle{Multi-parameter mechanism design and sequential
  posted pricing}. In \bibinfo{booktitle}{\emph{Proceedings of the 42nd ACM
  Symposium on Theory of Computing (STOC)}}. \bibinfo{pages}{311--320}.
\newblock
\urldef\tempurl%
\url{https://doi.org/10.1145/1806689.1806733}
\showDOI{\tempurl}


\bibitem[\protect\citeauthoryear{Cole and Rao}{Cole and Rao}{2017}]%
        {Cole2017}
\bibfield{author}{\bibinfo{person}{Richard Cole} {and} \bibinfo{person}{Shravas
  Rao}.} \bibinfo{year}{2017}\natexlab{}.
\newblock \showarticletitle{Applications of $\alpha$-Strongly Regular
  Distributions to Bayesian Auctions}.
\newblock \bibinfo{journal}{\emph{ACM Trans. Econ. Comput.}}
  \bibinfo{volume}{5}, \bibinfo{number}{4} (\bibinfo{year}{2017}),
  \bibinfo{pages}{18:1--18:29}.
\newblock
\urldef\tempurl%
\url{https://doi.org/10.1145/3157083}
\showDOI{\tempurl}


\bibitem[\protect\citeauthoryear{Cole and Roughgarden}{Cole and
  Roughgarden}{2014}]%
        {Cole2014a}
\bibfield{author}{\bibinfo{person}{Richard Cole} {and} \bibinfo{person}{Tim
  Roughgarden}.} \bibinfo{year}{2014}\natexlab{}.
\newblock \showarticletitle{The Sample Complexity of Revenue Maximization}. In
  \bibinfo{booktitle}{\emph{Proceedings of the 46th Annual ACM Symposium on
  Theory of Computing (STOC)}}. \bibinfo{pages}{243--252}.
\newblock
\urldef\tempurl%
\url{https://doi.org/10.1145/2591796.2591867}
\showDOI{\tempurl}


\bibitem[\protect\citeauthoryear{Correa, Foncea, Hoeksma, Oosterwijk, and
  Vredeveld}{Correa et~al\mbox{.}}{2017}]%
        {Correa2017}
\bibfield{author}{\bibinfo{person}{Jos{\'e} Correa}, \bibinfo{person}{Patricio
  Foncea}, \bibinfo{person}{Ruben Hoeksma}, \bibinfo{person}{Tim Oosterwijk},
  {and} \bibinfo{person}{Tjark Vredeveld}.} \bibinfo{year}{2017}\natexlab{}.
\newblock \showarticletitle{Posted Price Mechanisms for a Random Stream of
  Customers}. In \bibinfo{booktitle}{\emph{Proceedings of the 18th ACM
  Conference on Economics and Computation (EC)}}. \bibinfo{pages}{169--186}.
\newblock
\urldef\tempurl%
\url{https://doi.org/10.1145/3033274.3085137}
\showDOI{\tempurl}


\bibitem[\protect\citeauthoryear{Correa, Foncea, Hoeksma, Oosterwijk, and
  Vredeveld}{Correa et~al\mbox{.}}{2019a}]%
        {Correa:2019aa}
\bibfield{author}{\bibinfo{person}{Jose Correa}, \bibinfo{person}{Patricio
  Foncea}, \bibinfo{person}{Ruben Hoeksma}, \bibinfo{person}{Tim Oosterwijk},
  {and} \bibinfo{person}{Tjark Vredeveld}.} \bibinfo{year}{2019}\natexlab{a}.
\newblock \showarticletitle{Recent developments in prophet inequalities}.
\newblock \bibinfo{journal}{\emph{{ACM} {SIGecom} Exchanges}}
  \bibinfo{volume}{17}, \bibinfo{number}{1} (\bibinfo{year}{2019}),
  \bibinfo{pages}{61--70}.
\newblock
\urldef\tempurl%
\url{https://doi.org/10.1145/3331033.3331039}
\showDOI{\tempurl}


\bibitem[\protect\citeauthoryear{Correa, Foncea, Pizarro, and Verdugo}{Correa
  et~al\mbox{.}}{2019b}]%
        {Correa_2019}
\bibfield{author}{\bibinfo{person}{Jos{\'e} Correa}, \bibinfo{person}{Patricio
  Foncea}, \bibinfo{person}{Dana Pizarro}, {and} \bibinfo{person}{Victor
  Verdugo}.} \bibinfo{year}{2019}\natexlab{b}.
\newblock \showarticletitle{From Pricing to Prophets, and Back!}
\newblock \bibinfo{journal}{\emph{Operations Research Letters}}
  \bibinfo{volume}{47}, \bibinfo{number}{1} (\bibinfo{year}{2019}),
  \bibinfo{pages}{25--29}.
\newblock
\urldef\tempurl%
\url{https://doi.org/10.1016/j.orl.2018.11.010}
\showDOI{\tempurl}


\bibitem[\protect\citeauthoryear{Daskalakis and Weinberg}{Daskalakis and
  Weinberg}{2012}]%
        {Daskalakis2012b}
\bibfield{author}{\bibinfo{person}{Constantinos Daskalakis} {and}
  \bibinfo{person}{S.~Matthew Weinberg}.} \bibinfo{year}{2012}\natexlab{}.
\newblock \showarticletitle{Symmetries and Optimal Multi-dimensional Mechanism
  design}. In \bibinfo{booktitle}{\emph{Proceedings of the 13th ACM Conference
  on Electronic Commerce (EC)}}. \bibinfo{pages}{370--387}.
\newblock
\urldef\tempurl%
\url{https://doi.org/10.1145/2229012.2229042}
\showDOI{\tempurl}


\bibitem[\protect\citeauthoryear{Dhangwatnotai, Roughgarden, and
  Yan}{Dhangwatnotai et~al\mbox{.}}{2014}]%
        {Dhangwatnotai2014a}
\bibfield{author}{\bibinfo{person}{Peerapong Dhangwatnotai},
  \bibinfo{person}{Tim Roughgarden}, {and} \bibinfo{person}{Qiqi Yan}.}
  \bibinfo{year}{2014}\natexlab{}.
\newblock \showarticletitle{Revenue Maximization with a Single Sample}.
\newblock \bibinfo{journal}{\emph{Games and Economic Behavior}}
  \bibinfo{volume}{91}, \bibinfo{number}{C} (\bibinfo{year}{2014}),
  \bibinfo{pages}{318--333}.
\newblock
\urldef\tempurl%
\url{https://doi.org/10.1016/j.geb.2014.03.011}
\showDOI{\tempurl}


\bibitem[\protect\citeauthoryear{Dütting, Fischer, and Klimm}{Dütting
  et~al\mbox{.}}{2016}]%
        {Duetting2016}
\bibfield{author}{\bibinfo{person}{Paul Dütting}, \bibinfo{person}{Felix~A.
  Fischer}, {and} \bibinfo{person}{Max Klimm}.}
  \bibinfo{year}{2016}\natexlab{}.
\newblock \showarticletitle{Revenue Gaps for Discriminatory and Anonymous
  Sequential Posted Pricing}.
\newblock \bibinfo{journal}{\emph{CoRR}}  \bibinfo{volume}{abs/1607.07105v2}
  (\bibinfo{year}{2016}).
\newblock
\showeprint[arxiv]{1607.07105v2}


\bibitem[\protect\citeauthoryear{Ehsani, Hajiaghayi, Kesselheim, and
  Singla}{Ehsani et~al\mbox{.}}{2018}]%
        {Ehsani:2018aa}
\bibfield{author}{\bibinfo{person}{Soheil Ehsani},
  \bibinfo{person}{MohammadTaghi Hajiaghayi}, \bibinfo{person}{Thomas
  Kesselheim}, {and} \bibinfo{person}{Sahil Singla}.}
  \bibinfo{year}{2018}\natexlab{}.
\newblock \showarticletitle{Prophet Secretary for Combinatorial Auctions and
  Matroids}. In \bibinfo{booktitle}{\emph{Proceedings of the 29th Annual
  ACM-SIAM Symposium on Discrete Algorithms (SODA)}}.
  \bibinfo{pages}{700--714}.
\newblock
\urldef\tempurl%
\url{https://doi.org/10.1137/1.9781611975031.46}
\showDOI{\tempurl}


\bibitem[\protect\citeauthoryear{Ewerhart}{Ewerhart}{2013}]%
        {Ewerhart:2013aa}
\bibfield{author}{\bibinfo{person}{Christian Ewerhart}.}
  \bibinfo{year}{2013}\natexlab{}.
\newblock \showarticletitle{Regular Type Distributions in Mechanism Design and
  $\rho$-concavity}.
\newblock \bibinfo{journal}{\emph{Economic Theory}} \bibinfo{volume}{53},
  \bibinfo{number}{3} (\bibinfo{year}{2013}), \bibinfo{pages}{591---603}.
\newblock
\urldef\tempurl%
\url{https://doi.org/10.1007/s00199-012-0705-3}
\showDOI{\tempurl}


\bibitem[\protect\citeauthoryear{Feng, Hartline, and Li}{Feng
  et~al\mbox{.}}{2019}]%
        {Feng:2019aa}
\bibfield{author}{\bibinfo{person}{Yiding Feng}, \bibinfo{person}{Jason~D.
  Hartline}, {and} \bibinfo{person}{Yingkai Li}.}
  \bibinfo{year}{2019}\natexlab{}.
\newblock \showarticletitle{Optimal Auctions vs. Anonymous Pricing: Beyond
  Linear Utility}. In \bibinfo{booktitle}{\emph{Proceedings of the 20th ACM
  Conference on Economics and Computation (EC)}}. \bibinfo{pages}{885--886}.
\newblock
\urldef\tempurl%
\url{https://doi.org/10.1145/3328526.3329603}
\showDOI{\tempurl}


\bibitem[\protect\citeauthoryear{Fu, Immorlica, Lucier, and Strack}{Fu
  et~al\mbox{.}}{2015}]%
        {Fu2015}
\bibfield{author}{\bibinfo{person}{Hu Fu}, \bibinfo{person}{Nicole Immorlica},
  \bibinfo{person}{Brendan Lucier}, {and} \bibinfo{person}{Philipp Strack}.}
  \bibinfo{year}{2015}\natexlab{}.
\newblock \showarticletitle{Randomization Beats Second Price as a
  Prior-Independent Auction}. In \bibinfo{booktitle}{\emph{Proceedings of the
  16th {ACM} Conference on Economics and Computation (EC)}}.
  \bibinfo{pages}{323--323}.
\newblock
\urldef\tempurl%
\url{https://doi.org/10.1145/2764468.2764489}
\showDOI{\tempurl}


\bibitem[\protect\citeauthoryear{Giannakopoulos and Koutsoupias}{Giannakopoulos
  and Koutsoupias}{2018}]%
        {gk2014}
\bibfield{author}{\bibinfo{person}{Yiannis Giannakopoulos} {and}
  \bibinfo{person}{Elias Koutsoupias}.} \bibinfo{year}{2018}\natexlab{}.
\newblock \showarticletitle{Duality and Optimality of Auctions for Uniform
  Distributions}.
\newblock \bibinfo{journal}{\emph{SIAM J. Comput.}} \bibinfo{volume}{47},
  \bibinfo{number}{1} (\bibinfo{year}{2018}), \bibinfo{pages}{121--165}.
\newblock
\urldef\tempurl%
\url{https://doi.org/10.1137/16M1072218}
\showDOI{\tempurl}
\showeprint[arxiv]{1404.2329}


\bibitem[\protect\citeauthoryear{Giannakopoulos, Koutsoupias, and
  Lazos}{Giannakopoulos et~al\mbox{.}}{2017}]%
        {gkl2017}
\bibfield{author}{\bibinfo{person}{Yiannis Giannakopoulos},
  \bibinfo{person}{Elias Koutsoupias}, {and} \bibinfo{person}{Philip Lazos}.}
  \bibinfo{year}{2017}\natexlab{}.
\newblock \showarticletitle{Online Market Intermediation}. In
  \bibinfo{booktitle}{\emph{Proceedings of 44th International Colloquium on
  Automata, Languages, and Programming (ICALP)}}, Vol.~\bibinfo{volume}{80}.
  \bibinfo{pages}{47:1--47:14}.
\newblock
\urldef\tempurl%
\url{https://doi.org/10.4230/LIPIcs.ICALP.2017.47}
\showDOI{\tempurl}
\showeprint[arxiv]{1703.09279}


\bibitem[\protect\citeauthoryear{Giannakopoulos and Kyropoulou}{Giannakopoulos
  and Kyropoulou}{2017}]%
        {gkyr2015}
\bibfield{author}{\bibinfo{person}{Yiannis Giannakopoulos} {and}
  \bibinfo{person}{Maria Kyropoulou}.} \bibinfo{year}{2017}\natexlab{}.
\newblock \showarticletitle{The {VCG} Mechanism for Bayesian Scheduling}.
\newblock \bibinfo{journal}{\emph{ACM Trans. Econ. Comput.}}
  \bibinfo{volume}{5}, \bibinfo{number}{4} (\bibinfo{year}{2017}),
  \bibinfo{pages}{19:1--19:16}.
\newblock
\urldef\tempurl%
\url{https://doi.org/10.1145/3105968}
\showDOI{\tempurl}


\bibitem[\protect\citeauthoryear{Giannakopoulos and Zhu}{Giannakopoulos and
  Zhu}{2018}]%
        {gz2018}
\bibfield{author}{\bibinfo{person}{Yiannis Giannakopoulos} {and}
  \bibinfo{person}{Keyu Zhu}.} \bibinfo{year}{2018}\natexlab{}.
\newblock \showarticletitle{Optimal Pricing For {MHR} Distributions}. In
  \bibinfo{booktitle}{\emph{Proceedings of the 14th Conference on Web and
  Internet Economics (WINE)}}. \bibinfo{pages}{154--167}.
\newblock
\urldef\tempurl%
\url{https://doi.org/10.1007/978-3-030-04612-5_11}
\showDOI{\tempurl}
\showeprint[arxiv]{1810.00800v1}


\bibitem[\protect\citeauthoryear{Graham, Knuth, and Patashnik}{Graham
  et~al\mbox{.}}{1989}]%
        {Graham1989a}
\bibfield{author}{\bibinfo{person}{Ronald~L. Graham},
  \bibinfo{person}{Donald~E. Knuth}, {and} \bibinfo{person}{Oren Patashnik}.}
  \bibinfo{year}{1989}\natexlab{}.
\newblock \bibinfo{booktitle}{\emph{Concrete Mathematics: A Foundation for
  Computer Science}}.
\newblock \bibinfo{publisher}{Addison-Wesley Longman},
  \bibinfo{address}{Boston, MA, USA}.
\newblock


\bibitem[\protect\citeauthoryear{Hajiaghayi, Kleinberg, and
  Sandholm}{Hajiaghayi et~al\mbox{.}}{2007}]%
        {Hajiaghayi2007a}
\bibfield{author}{\bibinfo{person}{Mohammad~T. Hajiaghayi},
  \bibinfo{person}{Robert~D. Kleinberg}, {and} \bibinfo{person}{T. Sandholm}.}
  \bibinfo{year}{2007}\natexlab{}.
\newblock \showarticletitle{Automated Online Mechanism Design and Prophet
  Inequalities}. In \bibinfo{booktitle}{\emph{Proceedings of the 22nd
  Conference on Artificial Intelligence (AAAI)}}.
\newblock


\bibitem[\protect\citeauthoryear{Hardy, Littlewood, and Pólya}{Hardy
  et~al\mbox{.}}{1952}]%
        {Hardy:1952aa}
\bibfield{author}{\bibinfo{person}{G.~H. Hardy}, \bibinfo{person}{J.~E.
  Littlewood}, {and} \bibinfo{person}{G. Pólya}.}
  \bibinfo{year}{1952}\natexlab{}.
\newblock \bibinfo{booktitle}{\emph{Inequalities} (\bibinfo{edition}{2nd}
  ed.)}.
\newblock \bibinfo{publisher}{Cambridge University Press}.
\newblock


\bibitem[\protect\citeauthoryear{Hartline}{Hartline}{2013}]%
        {Hartlinea}
\bibfield{author}{\bibinfo{person}{Jason~D. Hartline}.}
  \bibinfo{year}{2013}\natexlab{}.
\newblock \bibinfo{title}{Mechanism Design and Approximation}.
  (\bibinfo{year}{2013}).
\newblock
\urldef\tempurl%
\url{http://jasonhartline.com/MDnA/}
\showURL{%
\tempurl}
\newblock
\shownote{Manuscript.}


\bibitem[\protect\citeauthoryear{Hartline and Roughgarden}{Hartline and
  Roughgarden}{2009}]%
        {Hartline2009a}
\bibfield{author}{\bibinfo{person}{Jason~D. Hartline} {and}
  \bibinfo{person}{Tim Roughgarden}.} \bibinfo{year}{2009}\natexlab{}.
\newblock \showarticletitle{Simple Versus Optimal Mechanisms}. In
  \bibinfo{booktitle}{\emph{Proceedings of the 10th ACM Conference on
  Electronic Commerce (EC)}}. \bibinfo{pages}{225--234}.
\newblock
\urldef\tempurl%
\url{https://doi.org/10.1145/1566374.1566407}
\showDOI{\tempurl}


\bibitem[\protect\citeauthoryear{Hill and Kertz}{Hill and Kertz}{1982}]%
        {Hill:1982aa}
\bibfield{author}{\bibinfo{person}{T.~P. Hill} {and} \bibinfo{person}{Robert~P.
  Kertz}.} \bibinfo{year}{1982}\natexlab{}.
\newblock \showarticletitle{Comparisons of Stop Rule and Supremum Expectations
  of I.I.D. Random Variables}.
\newblock \bibinfo{journal}{\emph{The Annals of Probability}}
  \bibinfo{volume}{10}, \bibinfo{number}{2} (\bibinfo{year}{1982}),
  \bibinfo{pages}{336--345}.
\newblock
\urldef\tempurl%
\url{https://doi.org/10.1214/aop/1176993861}
\showDOI{\tempurl}


\bibitem[\protect\citeauthoryear{Jin, Lu, Qi, Tang, and Xiao}{Jin
  et~al\mbox{.}}{2019a}]%
        {Jin:2019aa}
\bibfield{author}{\bibinfo{person}{Yaonan Jin}, \bibinfo{person}{Pinyan Lu},
  \bibinfo{person}{Qi Qi}, \bibinfo{person}{Zhihao~Gavin Tang}, {and}
  \bibinfo{person}{Tao Xiao}.} \bibinfo{year}{2019}\natexlab{a}.
\newblock \showarticletitle{Tight Approximation Ratio of Anonymous Pricing}. In
  \bibinfo{booktitle}{\emph{Proceedings of the 51st Annual ACM SIGACT Symposium
  on Theory of Computing (STOC)}}. \bibinfo{pages}{674--685}.
\newblock
\urldef\tempurl%
\url{https://doi.org/10.1145/3313276.3316331}
\showDOI{\tempurl}


\bibitem[\protect\citeauthoryear{Jin, Lu, Tang, and Xiao}{Jin
  et~al\mbox{.}}{2019b}]%
        {Jin2018}
\bibfield{author}{\bibinfo{person}{Yaonan Jin}, \bibinfo{person}{Pinyan Lu},
  \bibinfo{person}{Zhihao~Gavin Tang}, {and} \bibinfo{person}{Tao Xiao}.}
  \bibinfo{year}{2019}\natexlab{b}.
\newblock \showarticletitle{Tight Revenue Gaps among Simple Mechanisms}. In
  \bibinfo{booktitle}{\emph{Proceedings of the 30th Annual ACM-SIAM Symposium
  on Discrete Algorithms (SODA)}}. \bibinfo{pages}{209--228}.
\newblock
\urldef\tempurl%
\url{https://doi.org/10.1137/1.9781611975482.14}
\showDOI{\tempurl}


\bibitem[\protect\citeauthoryear{Kleinberg and Weinberg}{Kleinberg and
  Weinberg}{2012}]%
        {Kleinberg:2012}
\bibfield{author}{\bibinfo{person}{Robert Kleinberg} {and}
  \bibinfo{person}{Seth~Matthew Weinberg}.} \bibinfo{year}{2012}\natexlab{}.
\newblock \showarticletitle{Matroid Prophet Inequalities}. In
  \bibinfo{booktitle}{\emph{Proceedings of the 44th Annual ACM Symposium on
  Theory of Computing (STOC)}}. \bibinfo{pages}{123--136}.
\newblock
\urldef\tempurl%
\url{https://doi.org/10.1145/2213977.2213991}
\showDOI{\tempurl}


\bibitem[\protect\citeauthoryear{Mares and Swinkels}{Mares and
  Swinkels}{2011}]%
        {Mares:2011aa}
\bibfield{author}{\bibinfo{person}{Vlad Mares} {and} \bibinfo{person}{Jeroen~M.
  Swinkels}.} \bibinfo{year}{2011}\natexlab{}.
\newblock \showarticletitle{Near-optimality of Second Price Mechanisms in a
  Class of Asymmetric Auctions}.
\newblock \bibinfo{journal}{\emph{Games and Economic Behavior}}
  \bibinfo{volume}{72}, \bibinfo{number}{1} (\bibinfo{year}{2011}),
  \bibinfo{pages}{218--241}.
\newblock
\urldef\tempurl%
\url{https://doi.org/10.1016/j.geb.2010.08.008}
\showDOI{\tempurl}


\bibitem[\protect\citeauthoryear{Mares and Swinkels}{Mares and
  Swinkels}{2014}]%
        {Mares:2014aa}
\bibfield{author}{\bibinfo{person}{Vlad Mares} {and} \bibinfo{person}{Jeroen~M.
  Swinkels}.} \bibinfo{year}{2014}\natexlab{}.
\newblock \showarticletitle{On the analysis of asymmetric first price
  auctions}.
\newblock \bibinfo{journal}{\emph{Journal of Economic Theory}}
  \bibinfo{volume}{152} (\bibinfo{year}{2014}), \bibinfo{pages}{1--40}.
\newblock
\urldef\tempurl%
\url{https://doi.org/10.1016/j.jet.2014.03.010}
\showDOI{\tempurl}


\bibitem[\protect\citeauthoryear{Myerson}{Myerson}{1981}]%
        {Myerson1981a}
\bibfield{author}{\bibinfo{person}{Roger~B Myerson}.}
  \bibinfo{year}{1981}\natexlab{}.
\newblock \showarticletitle{Optimal Auction Design}.
\newblock \bibinfo{journal}{\emph{Mathematics of Operations Research}}
  \bibinfo{volume}{6}, \bibinfo{number}{1} (\bibinfo{year}{1981}),
  \bibinfo{pages}{58--73}.
\newblock
\urldef\tempurl%
\url{https://doi.org/10.1287/moor.6.1.58}
\showDOI{\tempurl}


\bibitem[\protect\citeauthoryear{Nisan}{Nisan}{2007}]%
        {Nisan2007a}
\bibfield{author}{\bibinfo{person}{Noam Nisan}.}
  \bibinfo{year}{2007}\natexlab{}.
\newblock \showarticletitle{Introduction to Mechanism Design (for Computer
  Scientists)}.
\newblock In \bibinfo{booktitle}{\emph{Algorithmic Game Theory}},
  \bibfield{editor}{\bibinfo{person}{Noam Nisan}, \bibinfo{person}{Tim
  Roughgarden}, \bibinfo{person}{{\'{E}}va Tardos}, {and}
  \bibinfo{person}{Vijay Vazirani}} (Eds.). \bibinfo{publisher}{Cambridge
  University Press}, Chapter~9.
\newblock


\bibitem[\protect\citeauthoryear{Olver, Lozier, Boisvert, and Clark}{Olver
  et~al\mbox{.}}{2010}]%
        {olveretal:10}
\bibfield{author}{\bibinfo{person}{Frank W.~J. Olver},
  \bibinfo{person}{Daniel~W. Lozier}, \bibinfo{person}{Ronald~F. Boisvert},
  {and} \bibinfo{person}{Charles~W. Clark}.} \bibinfo{year}{2010}\natexlab{}.
\newblock \bibinfo{booktitle}{\emph{NIST Handbook of Mathematical Functions}}.
\newblock \bibinfo{publisher}{Cambridge University Press}.
\newblock


\bibitem[\protect\citeauthoryear{Rudin}{Rudin}{1976}]%
        {rudin:76}
\bibfield{author}{\bibinfo{person}{Walter Rudin}.}
  \bibinfo{year}{1976}\natexlab{}.
\newblock \bibinfo{booktitle}{\emph{Principles of Mathematical Analysis}
  (\bibinfo{edition}{3rd} ed.)}.
\newblock \bibinfo{publisher}{McGraw-Hill}.
\newblock


\bibitem[\protect\citeauthoryear{Schweizer and Szech}{Schweizer and
  Szech}{2019}]%
        {SchweizerSzech2019}
\bibfield{author}{\bibinfo{person}{Nikolaus Schweizer} {and}
  \bibinfo{person}{Nora Szech}.} \bibinfo{year}{2019}\natexlab{}.
\newblock \showarticletitle{Performance Bounds for Optimal Sales Mechanisms
  Beyond the Monotone Hazard Rate Condition}.
\newblock \bibinfo{journal}{\emph{Journal of Mathematical Economics}}
  \bibinfo{volume}{82} (\bibinfo{year}{2019}), \bibinfo{pages}{202--213}.
\newblock
\urldef\tempurl%
\url{https://doi.org/10.1016/j.jmateco.2019.02.007}
\showDOI{\tempurl}


\bibitem[\protect\citeauthoryear{Yan}{Yan}{2011}]%
        {Yan2011}
\bibfield{author}{\bibinfo{person}{Qiqi Yan}.} \bibinfo{year}{2011}\natexlab{}.
\newblock \showarticletitle{Mechanism Design via Correlation Gap}. In
  \bibinfo{booktitle}{\emph{Proceedings of the 22nd Annual ACM-SIAM Symposium
  on Discrete Algorithms (SODA)}}. \bibinfo{pages}{710--719}.
\newblock
\urldef\tempurl%
\url{https://doi.org/10.1137/1.9781611973082.56}
\showDOI{\tempurl}


\end{thebibliography}

\appendix

\section{Order Statistics}
\label{append:prob_basics}
For the benefit of the reader and for ease of reference, in this section we collect
some fundamental facts from probability theory that are used in the paper.

Let $F$ be (the cumulative function (cdf) of) a continuous probability distribution
supported over an interval $D_F\subseteq [0,\infty)$. Then $F$ is an absolutely
continuous function and thus almost everywhere differentiable in $D_F$ with $F'(x)=f(x)$, where $f$
is the density function (pdf) of $F$. Let $X\sim F$ be a random variable drawn from
$F$. Then, for any $x\in D_F$, $\prob{X\leq x}=F(x)$

Let integers $1\leq k\leq n$. Then the pdf $f_{k:n}$ of $\orderstat{X}{k}{n}$, i.e.\
the $k$-th (lowest) order statistic out of $n$ i.i.d.\ draws from $F$, is given by
(see \citep[Eq.~(2.2.2)]{Arnold_2008})
$$
f_{k:n}(x)=k\binom{n}{k}F^{k-1}(x)(1-F(x))^{n-k} f(x)\;, \qquad x\in D_F\;.
$$
In particular, the cdf of $\orderstat{X}{n}{n}$ is simply $F_{n:n}(x)=F^n(x)$.

The \emph{exponential} distribution has cdf $F_{\mathcal E}(x)=1-e^{-y}$ and pdf
$f_{\mathcal E}(y)=e^{-y}$, for $y\in [0,\infty)$. In particular, for $Y\sim\mathcal
E$, the expected values of its order statistics are (see, e.g.,
\citep[Eq.~(4.6.6)]{Arnold_2008})
$$
\expect{\orderstat{Y}{k}{n}}=H_{n}-H_{n-k}\;,
$$
where $H_m=\sum_{i=1}^m\frac{1}{i}$ is the $m$-th harmonic number.

\section{Technical Lemmas}
\label{app:lemmas}

\begin{lemma}
\label{lem:harmonic_bounds_sequence}
The sequence $H_n-\ln (n)$ is strictly
decreasing for all integers $n\geq 1$, and converges (as $n\to\infty$) to the
Euler--Mascheroni constant $\gamma\approx 0.577$.
\end{lemma}
\begin{proof}
This is a well-known fact from analysis, see e.g.~\citep[Eq.~(6.64)]{Graham1989a}.
\end{proof}

\begin{lemma}\label{lem:specific_price}
\label{lem:bound_pricing_seq_parameter} For all integers $n\geq 5$, $$
\ln(n)-\ln(\ln(n)) \leq H_n-1\;. $$
\end{lemma}

\begin{proof}
Since $\ln n$ is increasing with respect to the positive integer $n$, using the
value of the constant $\gamma\approx 0.577$ it is not difficult to numerically check
that for any $n\geq 5$ we have
$$
\ln n \geq \ln 5 \approx 1.609 > 1.526\approx  e^{1-\gamma}\;.
$$
As a result, using \cref{lem:harmonic_bounds_sequence}, we can see that
$$
H_n-1 \geq \ln(n)+\gamma-1
=\ln \frac{n}{e^{1-\gamma}}
> \ln \frac{n}{\ln n}
=\ln n - \ln \ln n\;.
$$
\end{proof}

\begin{lemma}
\label{lem:helper_bound_1}
For any integer $n\geq 2$ and any $x\in[1/e,1]$,
$$
\ln(1/x)x(1-x)^{n-2} \leq \frac{1}{e}\left(\frac{e-1}{e}\right)^{n-2}\;.
$$
\end{lemma}
\begin{proof}
If we define the function $f:(0,1]\map (0,\infty)$ with
$$
f_n(x)= \ln\left(\frac{1}{x}\right)x(1-x)^{n-2}\;,
$$
then $f_n(1/e)=\frac{1}{e}\left(\frac{e-1}{e}\right)^{n-2}$. Thus, it is enough to
show that $f_n$ is monotonically decreasing in $[1/e,1]$. Taking its derivative, we
see that indeed
\begin{align*}
f_n'(x) &= (1-x)^{n-3} \left[x-1+(1-(n-1)x) \ln\left(\frac{1}{x}\right)\right]\\
    &\leq (1-x)^{n-3} \left[x-1+(1-x) \ln\left(\frac{1}{x}\right)\right]\\
    &= -(1-x)^{n-2} \left[1- \ln\left(\frac{1}{x}\right)\right]\\
    &\leq 0\;.
\end{align*}
The first inequality is due to the fact that $n\geq 2$ and $\ln(1/x)\geq 0$ and the
last one due to $x\geq 1/e$.
\end{proof}

\begin{lemma}
\label{lem:lambdastar}
The quantity
\[\xi(\lambda)=1-\left(1+\frac{\Gamma(2-\lambda)^{-1/\lambda}}{\lambda}\right)e^{-\Gamma(2-\lambda)^{-1/\lambda}}\]
has a unique zero $\lambda^\ast$ over the interval $(0,1]$; moreover,
$\xi(\lambda)<0$ for $0<\lambda<\lambda^\ast$, and  $\xi(\lambda)>0$ for
$\lambda^\ast<\lambda\leq 1$.
\end{lemma}

\begin{proof} For the proof of this result we will use the well-known fact that the
gamma function is continuously differentiable over $[1,2]$, and its derivative
changes sign from negative to positive at a single point in this interval
(\cite[Ch. 5]{olveretal:10}).

Observe that $\xi$ is continuously differentiable in $(0,1]$, with
\[\xi'(\lambda)=\frac{e^{-\Gamma(2-\lambda)^{-1/\lambda}}}{\lambda^2
\Gamma(2-\lambda)^{1+1/\lambda}}\left(\Gamma(2-\lambda)^{-1/\lambda}+\lambda-1\right)\Gamma'(2-\lambda)\;;\]
the first two factors in this expression are strictly positive over $(0,1]$, whereas
the third factor changes sign from positive to negative at a unique point, say
$\tilde{\lambda}$. It follows that $\xi(\lambda)$ is strictly increasing over
$(0,\tilde{\lambda})$ and strictly decreasing over $(\tilde{\lambda},1]$. As
$\xi(0^+)=-\infty$ and $\xi(1)>0$, we get the desired result.
\end{proof}

\section{Analytic Properties of Auxiliary Functions}
\label{app:auxfunprops}

\begin{lemmanonum}[\cref{lem:propsg}]
 Functions $g_n$ defined in
\eqref{eq:gdef} have a unique point of maximum
$$\xi_n= \argmax_{c\geq 0} g_n(c)\;.$$ 
Furthermore, for all $n\geq 17$,
$$
\xi_n \leq 1\;.
$$
\end{lemmanonum}

\begin{proof}

Let $n$ be a positive integer. First we compute the first and second derivatives of
$g_n$:
\begin{align}
g_n'(c) &= 1 - \left(1-e^{-c(H_n-1)} \right)^n\left(1+n\frac{c(H_n-1)}{e^{c(H_n-1)}-1}\right)\label{eq:helper_g_deriv}\\
g_n''(x) &= \frac{n(H_n-1) \left(1-e^{-c(H_n-1)}\right)^n \left[e^{c(H_n-1)} (c(H_n-1)-2)-n c(H_n-1)+2\right]}{\left(e^{c(H_n-1)}-1\right)^2}\;.\notag
\end{align}

There is a unique point $\tau_n>0$ on which function $x\mapsto e^x (x-2)-n x+2$
changes sign from negative to positive, so the same holds for function $g_n''$;
thus, $g_n'$ is strictly decreasing over $(0,\tau_n)$ and increasing over
$(\tau_n,\infty)$. Furthermore, notice that $\lim_{c\to 0} g_n'(c)=1$ and
$\lim_{c\to\infty}g_n'(c)=0$. This means that there has to be a unique point
$\xi_n>0$ (where $\xi_n < \tau_n$) at which $g'_n$ changes sign from positive to
negative. Therefore, $\xi_n$ is the unique global maximizer of $g_n$.

In order to prove that a positive real $y$ is above that maximization point, i.e.,
$y\geq \xi_n$, it is enough to show that function $g_n$ is decreasing at $y$, or
equivalently, $g_n'(y)<0$. So, recalling from
\cref{lem:harmonic_bounds_sequence} that $H_n-1>\ln(n)+\gamma-1$, in order to
complete our proof it is enough to show that
$$
g_n'\left(\frac{\ln n +\gamma-1}{H_n-1}\right) < 0 \qquad \text{for all}\;\; n\geq 17\;.
$$
We will do this by demonstrating that $g_n'\left(\frac{\ln n +\gamma-1}{H_n-1}\right)$ is a decreasing
sequence (with respect to $n$) that gets negative for $n=17$. Using the explicit
formula \eqref{eq:helper_g_deriv} for $g_n'$, we can see that
$$
g_n'\left(\frac{\ln n +\gamma-1}{H_n-1}\right)=1-\left(1-\frac{e^{1-\gamma}}{n}\right)^n\left(1+\frac{\ln n+\gamma -1}{e^{\gamma -1}-1/n} \right)
$$
is decreasing, since sequences $\left(1-\frac{e^{1-\gamma}}{n}\right)^n$ and
$1+\frac{\ln n+\gamma -1}{e^{\gamma -1}-1/n}$ are both positive and increasing.
Finally it is easy to compute (by simple substitution) that for $n=17$,
$g_{17}'\left(\frac{\ln 17 +\gamma-1}{H_{17}-1}\right)\approx -0.019<0$.
\end{proof}

\begin{lemma}
\label{lem:3}
For the functions $g_n$ defined in \eqref{eq:gdef},
\begin{equation}
\max_{c\in[0,1]} g_n(c) =1-O\left(\frac{\ln\ln n}{\ln n}\right)=1-o(1)\;.
\end{equation}
\end{lemma}
\begin{proof}
For any integer $n\geq 5$, we will lower bound the maximum value of $g_n(x)$ by
evaluating it on the following numbers:
$$
\tilde c_n = \frac{\ln n-\ln\ln n}{H_n-1}=\frac{\ln\left(n/\ln n\right)}{H_n-1} \leq 1\; ;
$$
the inequality holds due to~\cref{lem:bound_pricing_seq_parameter}. We have that

\begin{align*}
g_n(\tilde c_n) &= \frac{\ln\left(n/\ln n\right)}{H_n-1}\left[1-\left(1-e^{-\ln\left(\frac{n}{\ln n}\right)}\right)^n\right]\\
    &=\frac{\ln n-\ln\ln n}{H_n-1}\left[1-\left(1-\frac{\ln n}{n}\right)^n\right]\\
    & \geq \frac{\ln n-\ln\ln n}{\ln n}\left(1-\frac{1}{n}\right)\\
    & = 1-O\left(\frac{\ln\ln n}{\ln n}\right)\;,
\end{align*}
the inequality being a consequence of the fact that
$\left(1-\frac{x}{n}\right)^n\leq e^{-x}$ for any positive real $x\leq n$ with $x\gets
\ln n$ and also due to $H_n\leq 1+\ln n$.
\end{proof}

\begin{lemmanonum}[\cref{lem:basicpareto}]
For the rescaled Pareto distribution $F_{\lambda,r}$ given by \eqref{eq:pareto_instances}, for $\lambda\in(0,1]$ and $r\in(0,1]$, we have the following expressions for its expected second-highest order statistic and optimal anonymous pricing,
\begin{enumerate}
    \item $\expect[X\sim F_{\lambda,r}]{X_{n-1:n}}=1+r\left(\beta_{n,\lambda}-1\right)$;
    \item $\price{F_{\lambda,r}}{n}=\sup_{0\leq q\leq 1}\left(1+r\left(\frac{1}{q^\lambda}-1\right)\right)\left(1-(1-q)^n\right)$,
\end{enumerate}
where $\beta_{n,\lambda}$ is given by \eqref{eq:helperreg_ab}.
\end{lemmanonum}

\begin{proof}
Point 2 follows directly from the change of variables
\[q=1-F_{\lambda,r}(x)\quad\Leftrightarrow\quad x=1+r\left(\frac{1}{q^\lambda}-1\right)=1-r+\frac{r}{q^\lambda}\;.\]

To prove point 1, we use the same change of variables to express the expectation in terms of the gamma function:
\begin{align*}
    \expect{X_{n-1:n}}&=n(n-1)\int_1^\infty x F^{n-2}(x)(1-F(x))\,dF(x)\\
    &=n(n-1)\int_0^1\left(1-r+\frac{r}{q^\lambda}\right)q(1-q)^{n-2}\,dq\\
    &= (1-r)n(n-1)\mathrm{B}(2,n-1)+rn(n-1)\mathrm{B}(2-\lambda,n-1)\\
    &=1-r+r\frac{n!\Gamma(2-\lambda)}{\Gamma(n+1-\lambda)}\\
    &=1+r(\beta_{n,\lambda}-1)\;.
\end{align*}
\end{proof}

\begin{lemmanonum}[\cref{lem:supglam}] For each $\lambda\in(0,1]$, let $g_\lambda$ be defined as in \eqref{eq:glambda} and $\eta(\lambda)$ be the unique positive solution of the equation $e^x=1+\frac{x}{\lambda}$. Let also $\lambda^\ast$ be the unique root of \eqref{eq:lambdastar}. The supremum of the function $g_\lambda$ over $[0,1]$ is as follows.
\begin{align*}
\text{If }0<\lambda\leq\lambda^\ast\text{, then}\quad&\sup_{c\in[0,1]}g_\lambda(c)=\sup_{c\geq 0}g_\lambda(c)=\frac{\eta(\lambda)^{1-\lambda}}{\Gamma(2-\lambda)(\lambda+\eta(\lambda))}\;.\\
\text{If }\lambda^\ast\leq\lambda\leq 1\text{, then}\quad&\sup_{c\in[0,1]}g_\lambda(c)=\phantom{\sup_{c\geq0}}g_\lambda(1)=1-e^{-\Gamma(2-\lambda)^{-1/\lambda}}\;.
\end{align*}
\end{lemmanonum}

\begin{proof} Fix $\lambda\in(0,1)$. Observe that $g_\lambda$ is twice continuously
differentiable; with some calculations one can see that
\begin{align*}
g'_\lambda(c)&=1-e^{-(c\Gamma(2-\lambda))^{-1/\lambda}}\left[1+\frac{\left(c\Gamma(2-\lambda)\right)^{-1/\lambda}}{\lambda}\right];\\
g''_\lambda(c)&=\frac{e^{-(c\Gamma(2-\lambda))^{-1/\lambda}}}{\lambda^2c^{1+2/\lambda}\Gamma(2-\lambda)^{2/\lambda}}\left[-1+(1-\lambda)(c\Gamma(2-\lambda))^{1/\lambda}\right]\;;
\end{align*} one also has the limits $g_\lambda(0)=0$, $g_\lambda(\infty)=0$,
$g'_\lambda(0)=1$, $g'_\lambda(\infty)=0$, $g''_\lambda(0)=0$,
$g''_\lambda(\infty)=0$. Now note that there is a unique point,
$\tau_\lambda=\frac{1}{\Gamma(2-\lambda)(1-\lambda)^\lambda}$, at which
$g''_\lambda$ changes sign from negative to positive, so that there is a unique
point $\xi_\lambda\leq\tau_\lambda$ at which $g'_\lambda$ changes sign from positive
to negative. Thus $g_\lambda$ has a unique peak over $[0,\infty)$ which corresponds
to the unique root of its derivative. This
peak occurs in the interval $[0,1]$ if and only if $g'_\lambda(1)\leq 0$; solving
for $\lambda$, we get $\lambda\leq\lambda^\ast$ as in~\eqref{eq:lambdastar} and
\cref{lem:lambdastar}. Thus, if $\lambda\geq\lambda^\ast$, the supremum of
$g_\lambda(c)$ over $[0,1]$ is achieved at $c=1$; this includes the case $\lambda=1$ since a similar analysis yields that $g''_\lambda$ is strictly negative and $g'_\lambda$ is strictly positive. On the other hand, for $\lambda\leq\lambda^\ast$,
by looking at the equation $g'_\lambda(c)=0$ and performing the change of variables
$x=(c\Gamma(2-\lambda))^{-1/\lambda}$ we get $e^x=1+\frac{x}{\lambda}$, or
$x=\eta(\lambda)$; plugging these back into $g_\lambda(c)$ gives us the desired
result.
\end{proof}

\begin{lemma}\label{lem:lowboundtechnical1} For each $\lambda\in(0,1]$, we have the
limit
\[\adjustlimits\lim_{n\rightarrow\infty}\sup_{c\geq1/(\Gamma(2-\lambda)n^\lambda)}c\left[1-\left(1-\frac{(c\Gamma(2-\lambda))^{-1/\lambda}}{n}\right)^n\right]
=\sup_{c\geq
0}g_\lambda(c)\;.\]
\end{lemma}

\begin{proof} Define the sequence of auxiliary functions $h_{n,\lambda}$ for $n\geq
2$,

$$
h_{n,\lambda}(c)=
\begin{cases}
c, & c\leq
\frac{1}{\Gamma(2-\lambda)n^\lambda}\;,\\
c\left[1-\left(1-\frac{(c\Gamma(2-\lambda))^{-1/\lambda}}{n}\right)^n\right], & c\geq\frac{1}{\Gamma(2-\lambda)n^\lambda}\;.
\end{cases}
$$
These functions are continuous, vanish at infinity, and converge pointwise to
$g_\lambda$, which is also continuous. Also, for any $\epsilon>0$, the restriction
of $h_{n,\lambda}$ to $[\epsilon,\infty)$ forms an eventually decreasing sequence (with respect to $n$) ,
since $x\leq n\leq m$ implies $(1-x/n)^n\leq (1-x/m)^m$ (\cite[Theorem~35]{Hardy:1952aa}). Thus, by Dini's Theorem \cite[Theorem~7.13]{rudin:76}, over any
interval $[a,b]$ with $0<a<b<\infty$, the sequence $h_{n,\lambda}$ converges
uniformly to $g_\lambda$ and we have
$\adjustlimits\lim_{n\rightarrow\infty}\sup_{c\in[a,b]}h_{n,\lambda}(c)=\sup_{c\in[a,b]}g_\lambda(c)$.

To conclude the proof, let $z_\lambda^\ast$ be the unique maximizer of
$g_\lambda(c)$ (see also the proof of \cref{lem:supglam}). Take $0<\epsilon\leq
g_\lambda(z^\ast)$. Take $N$ large enough so that $h_{n,\lambda}$ are decreasing
over $[\epsilon,\infty)$, for $n\geq N$. Take $\delta$ such that $h_{N}(c)\leq
\epsilon$ for all $c\geq\delta$. Putting all this together, we have
\begin{align*}
\adjustlimits
\limsup_{n\rightarrow\infty}\sup_{c\in[0,\epsilon]}h_{n,\lambda}(c) &\leq\epsilon\;;\\
\limsup_{n\rightarrow\infty}\sup_{c\in[\epsilon,\delta]}h_{n,\lambda}(c) &=\sup_{c\in[\epsilon,\delta]}g_{\lambda}(c)\leq
g_\lambda(z^\ast)\;;\\
\limsup_{n\rightarrow\infty}\sup_{c\geq\delta}h_{n,\lambda}(c) &\leq
\epsilon\;,
\end{align*}
and thus $\limsup_{n\rightarrow\infty}\sup_{c\geq
0}h_{n,\lambda}(c)\leq\sup_{c\geq 0}g_\lambda(c)$. Since by elementary analysis we
also have $\liminf_{n\rightarrow\infty}\sup_{c\geq
0}h_{n,\lambda}(c)\geq\sup_{c\geq 0}g_\lambda(c)$, this concludes the proof.
\end{proof}

\begin{lemma}
\label{lem:lowboundtechnical2} 
For each $\lambda\in(0,1]$ and $n\geq 2$, the function
\[H_{n,\lambda}(1,q)=\frac{1-(1-q)^n}{q^\lambda}\;,\] defined over
$q\in[0,1]$, has a unique maximizer $\xi_{n,\lambda}$. Moreover, if in addition
$\lambda^\ast<\lambda\leq 1$, then
$\xi_{n,\lambda}\leq\beta_{n,\lambda}^{-1/\lambda}$ for large enough $n$. Here
$\lambda^\ast$ is the unique root of \eqref{eq:lambdastar} and $\beta_{n,\lambda}$
is defined as in \eqref{eq:helperreg_ab}.
\end{lemma}

\begin{proof} For simplicity we shall assume that $\lambda<1$; when $\lambda=1$ a
similar analysis implies that the unique maximizer occurs at $\xi_{n,\lambda}=0$.

The first derivative of $H_{n,\lambda}$ is
\[\frac{\partial\,
H_{n,\lambda}(1,q)}{\partial\,q}=\frac{\lambda}{q^{1+\lambda}}\left[(1-q)^{n-1}\left[1+q\left(\frac{n}{\lambda}-1\right)\right]-1\right]\;,\]
which (for $\lambda<1$) has a positive pole at $q=0$ and is negative at $q=1$. Its
first factor is always positive, and its second factor can be differentiated:
\[\tilde{H}_{n,\lambda}(1,q)=(1-q)^{n-1}\left[1+q\left(\frac{n}{\lambda}-1\right)\right]-1\;;\]
\[\frac{\partial\,\tilde{H}_{n,\lambda}(1,q)}{\partial\, q}=\frac{n}{\lambda}(1-q)^{n-2}\left[(1-\lambda)-q(n-\lambda)\right]\;.\]
This function changes sign from positive to negative at a single point,
$\tau_{n,\lambda}=\frac{1-\lambda}{n-\lambda}$. Thus, both $\tilde{H}_{n,\lambda}$
and $\frac{\partial}{\partial\, q} H_{n,\lambda}$ change sign from positive to negative at a single
point $\xi_{n,\lambda}\geq\tau_{n,\lambda}$, which enables us to conclude that
$H_{n,\lambda}(1,q)$ has a single peak.

Next, we argue that the unique maximizer of $H_{n,\lambda}(1,q)$ is smaller than the
quantity $\beta_{n,\lambda}^{-1/\lambda}$, for large enough $n$. In order to do
this, observe that
\[\left.\frac{\partial\,
H_{n,\lambda}\left(1,q\right)}{\partial\,q}\right|_{q=\beta_{n,\lambda}^{-1/\lambda}}=\lambda\beta_{n,\lambda}^{1+1/\lambda}\left[\left(1-\beta_{n,\lambda}^{-1/\lambda}\right)^{n-1}\left[1+\beta_{n,\lambda}^{-1/\lambda}\left(\frac{n}{\lambda}-1\right)\right]-1\right]\;.\]

The first factor of the above expression is strictly positive for all choices of $n$
and $0<\lambda\leq1$. The second factor actually has a limit as
$n\rightarrow\infty$; it converges to

\[\left(1+\frac{\Gamma(2-\lambda)^{-1/\lambda}}{\lambda}\right)e^{-\Gamma(2-\lambda)^{-1/\lambda}}-1\;.\]

This quantity is negative precisely when $\lambda>\lambda^\ast$, as proven in
\cref{lem:lambdastar}. Thus, as long as $\lambda>\lambda^\ast$, we have that
$\frac{\partial}{\partial\, q} H_{n,\lambda}\left(1,\beta_{n,\lambda}^{-1/\lambda}\right)$ is negative
for large enough $n$, which implies that the unique maximizer of
$H_{n,\lambda}(1,q)$ is to the left of $\beta_{n,\lambda}^{-1/\lambda}$.
\end{proof}
\end{document}